\documentclass[preprint]{sigplanconf}

\pdfoutput=1
\usepackage{amsmath}
\usepackage{amssymb}
\usepackage{url}
\usepackage{epsfig}

\newcommand{\semb}{[ \! [}
\newcommand{\seme}{] \! ]}
\newtheorem{theo}{Theorem}
\newtheorem{defi}[theo]{Definition}
\newtheorem{deflem}[theo]{Definition \& Lemma}
\newtheorem{lemm}[theo]{Lemma}
\newtheorem{prop}[theo]{Proposition}

\newtheorem{exem}[theo]{Example}

\def\norm#1{\mbox{$\| #1 \|$}}
\def\normA#1{\mbox{${\| #1 \|}_A$}}

\def\normL#1{\mbox{${\| #1 \|}_L$}}

\def\stackreh#1#2{\mathrel{\mathop{#1}\limits_{#2}}}
\def\argmin#1#2{\stackreh{\mbox{argmin }}{\tiny #1 \wedge #2 \leq \alpha \leq #1 \vee #2} |\alpha|}
\def\argmax#1#2{\stackreh{\mbox{argmax }}{\tiny #1 \wedge #2 \leq \alpha \leq #1 \vee #2} |\alpha|}
\def\im{\mbox{Im }}
\def\RAff{{\Bbb A \Bbb R}}
\def\RA{{\Bbb A}}
\def\R{{\Bbb R}}
\def\I{{\Bbb I \Bbb R}}

\def\var{\mbox{Var}}

 \newcommand\ForAuthors[1]
 {\par\smallskip                     
  \begin{center}
   \fbox
   {\parbox{0.9\linewidth}
    {\raggedright\sc--- #1}
   }
  \end{center}
  \par\smallskip                     %
 }        

\begin{document}
\bibliographystyle{plain}

\copyrightyear{2009}
\copyrightdata{}
\title{Perturbed affine arithmetic for invariant computation in
              numerical program analysis\thanks{This material is based upon work supported
ANR project EvaFlo, and ITEA 2 project ES\_PASS.}}
\authorinfo{Eric Goubault\and Sylvie Putot}
{CEA LIST, Modelisation and Analysis of Systems in Interaction}
{\{eric.goubault,sylvie.putot\}@cea.fr}

\maketitle
\begin{abstract}
We completely describe a new domain for abstract interpretation of 
numerical programs. Fixpoint iteration in this domain is {\em proved} to converge to finite precise
invariants for (at least) the class of {\em stable linear recursive filters of
any order}. Good evidence shows it behaves well also for some non-linear
schemes. The result, and the structure of the domain, rely on an interesting interplay between order
and topology. 
\end{abstract}
\category{D.2.4}{Software/Program Verification}
[Validation]
\category{F.3.1}{Specifying and Verifying and Reasoning about Programs}
[Mechanical verification]
\category{F.3.2}{Semantics of Programming Languages}
[Program analysis]
\terms
Theory, Verification
\keywords
Abstract interpretation, numerical programs

\section{Introduction}

An everlasting challenge of the verification of programs involving
numerical computations is to efficiently find accurate invariants for values of variables. 
Even though machine computations use finite precision arithmetic, it is important to rely on the 
properties of real numbers and estimate the real number values of the program variables first, before even trying to characterize the floating-point number
invariants. 
We refer the reader to \cite{SAS06}, which describes a way to go from this to floating-point analysis, or to the static linearization techniques
of \cite{Mine}.

In \cite{SAS06} as well, 
some first ideas about an abstract interpretation
domain which would be expressive enough for deriving these invariants, 
were sketched. It relied on a more 
accurate alternative to interval arithmetic: affine arithmetic,   
the concretization of which is a center-symmetric polytope. But, contrarily to
existing numerical relational 
abstract domains with polyhedral concretization (polyhedra \cite{Polyedres}
of course, but also zones, octagons \cite{Octagon} etc.), 
dependencies in affine arithmetic are implicit, making the semantics very economical.
Also, affine arithmetic is close to Taylor models, which can be exploited to give precise abstractions 
of non-linear computations.

But these advantages are at a theoretical cost: the partial order and the correctness of the abstract 
computations are intricate to find and prove.
In this article, we construct a ``quasi'' lattice abstract domain, and study the convergence of fixpoint computations.
We show how the result of the join operators we define can be considered as a perturbation of the affine forms, 
and thus how the fixpoint iteration can be seen as a perturbation of the numerical schemes we analyze.
A crucial point is that our abstract domain is both almost a bounded complete lattice, and an ordered Banach space,
 where approximation theorems and convergence properties of numerical schemes naturally fit.
As an application of the framework, we prove that our approach allows us to accurately bound the values 
of variables for stable linear recursive filters of any order.

\paragraph{Contributions}
This article fully describes a general ``completeness'' result of the abstract domain, 
for a class of numerical programs (linear recursive filters of any
order), meaning that we prove that the abstract
analysis results will end up with finite numerical bounds whenever the 
numerical scheme analyzed has this property. We also show good evidence
that, on this class of programs, we can get as close an over-approximation of the real
result as we want. 

The abstract domain on which we prove this result is a generalization of the
one of \cite{SAS06}; better join and meet operators are described,
and the full order-theoretic structure is described (sketches of proofs
are given). 

A new feature of this domain, with respect to the other
numerical abstract domains, is that it does not only have an order-theoretic
structure, but also a topological one, the interaction of which 
plays an important role in our results. The domain is an ordered Banach
space, ``almost'' a Riesz space, which are structures of interest in 
functional analysis and optimization theory. This is not just a coincidence: 
correctness of the abstraction relies on the correctness of functional
evaluations in the future, i.e. continuations. This opens up promises
for useful generalizations and new techniques for solving the 
corresponding semantic equations. 

\paragraph{Contents}
Section \ref{sec_problem} introduces the general problematic of finding
precise invariants for numerical programs, and defines an interesting
sub-class of problems, that is linear recursive filters of any order.
We also introduce the classical affine forms \cite{Stolfi} introduced in numerical
mathematics, on which our work elaborates. 

Section \ref{poBanach} extends these affine forms to deal with 
{\em static analysis invariants}. We show that the set of such generalized
forms has the structure of an {\em ordered Banach space}, which almost has least
upper bounds and greatest lower bounds: it actually only has {\em
maximal lower bounds} and {\em minimal upper bounds}, in general. An equivalent
of {\em bounded completeness} is proved using the interplay between
the partial order and the topology (from the underlying Banach space). 

We develop particular Kleene iteration techniques  in 
Section \ref{convsection}. With these, we prove that we can find
finite bounds for the invariants of {\em stable} linear recursive filters.
We also show evidence that these abstractions give good result in practice,
using our current C implementation of the abstract semantics. 

Finally, we give hints about current and future work in Sections \ref{improv} 
and \ref{conclusion}.

\section{Problematic}
\label{sec_problem}
We are interested in {\em numerical schemes} in the large. This includes
signal processing programs, control programs such as the ones used in
aeronautics, automotive and space industry, libraries for computing
transcendental functions, and as a long-term goal, simulation programs (including
the solutions of ordinary or partial differential equations). The context of our work 
is the determination of the
accuracy reached by the finite precision (generally IEEE 754) implementation of these 
numerical schemes, see for instance \cite{SAS01,ESOP02,Dag03,SAS06,FMICS07}. 
But it is already a difficult problem
for these numerical programs, to determine run-time errors (RTEs) statically,
just because the bounds of the results of numerical computation are hard to find. 
These bounds are not only hard to find for floating-point arithmetic, but also for real 
arithmetic, which is the first critical step towards solving the complete problem.

{\em In the sequel, 
we are describing a precise 
abstract domain of affine forms for bounding real number calculations, 
in the sense of abstract
interpretation \cite{Cousot}.}

We give in Section \ref{numschemes} a class of simple programs that
are pervasive in the field of numerical computing: {\em linear 
recursive filters
of any order}.
They are encountered generally in signal processing and control programs,
but encompass also all linear recurrence schemes that can be found in 
simulation programs. We will study {\em extensively} the behavior
of our abstract domain on such programs.

Of course, we are also interested in non-linear schemes, and already studied some coming for example 
from the solution of a conjugate gradient algorithm, or algorithms for estimating transcendental 
functions.
And, as we will see in the description of our abstract semantics, one of
the interesting points of affine forms is that they behave well also for
non-linear computations, in a much more precise and natural way than
with classical polyhedra, or zones/octagons. But we have not reached
yet the point in the theory where we can state as precise statements as
for the analysis of {\em linear} dynamical systems, although strong 
practical evidence show that our method gives very good results as well
for some non-linear dynamical systems (see for instance \cite{SAS06, FMICS07}
for some examples that were already solved with a much coarser abstract
domain than the one of this article).

\subsection{A class of numerical schemes of interest}

\label{numschemes}

Let us consider the following class of
program, that we will study in depth with our abstract domain in Section \ref{linear}:

\begin{verbatim}
filter(float x[n+1]) {
  real e[n+1];
  e[*] = input(m,M); [1]
  while (true) { [2]
    e[n+1] = input(m,M);
    x[n+1] = a[1]*x[1]+a[2]*x[2]+...+a[n]*x[n]
     + b[1]*e[1]+b[2]*e[2]+...+b[n+1]*e[n+1]; [3]
    x[n] = x[n+1]; ... x[1] = x[2]; [4] } }
\end{verbatim}

In the program above, \texttt{a[]} is an array of $n$ constants $a_i$,
$i=1,\ldots,n$ (indices of arrays start at 1),
\texttt{b[]} is an array of $n+1$ constants $b_i$, $i=1,\ldots, n+1$.
\texttt{M} and \texttt{m} are parameters, giving the bounds $M$ and $m$
of the successive inputs over time.
For purposes of simplicity, as was discussed in the introduction, types
of variables are {\em real} number types. We use the notation 
\texttt{e[*]=input(m,M);} to denote the sequence of $n+1$ input assignments 
between $m$ and $M$. 
At iterate $k$ of the filter, variable
\texttt{x[i]} represents the value $x_{k+i}$ of the output. Our main interest here
is in the {\em invariant} at control point [2] (control points are
indicated as numbers within square brackets).
 
The program \texttt{filter} describes the infinite iteration of a filter of order $n$ with
coefficients $a_1,\ldots,a_n$, $b_1,\ldots,b_{n+1}$ 
and a new input $e$ between $m$ and $M$ at each iteration :
\begin{equation}
\label{LRF}
x_{k+n+1} = \sum_{i=1}^{n} a_i x_{k+i} + \sum_{j=1}^{n+1} b_j e_{k+j},
\end{equation}
starting with initial conditions $x_1,\ldots, x_n, x_n$. 

We rewrite (\ref{LRF}) as:
\begin{equation}
\label{matrixeq}
X_{k+1}=A X_k + B E_{k+1},
\end{equation}
with $$
\begin{array}{cc}
X_k=\left(\begin{array}{c} 
x_{k} \\
x_{k+1} \\
\ldots \\
x_{k+n} \\
\end{array}\right)
& 
E_{k+1}=\left(\begin{array}{c}
e_{k} \\
e_{k+1} \\
\ldots \\
e_{k+n+1} \\
\end{array}\right)
\end{array}
$$
$$\begin{array}{cc}
A=\left(\begin{array}{cccc}
0 & 1 & \ldots & 0 \\
0 & 0 & \ldots & 0 \\
\ldots \\
0 & 0 & \ldots & 1 \\
a_1 & a_2 & \ldots & a_n
\end{array}\right) &
B = \left(\begin{array}{cccc}
0 & 1 & \ldots & 0 \\
0 & 0 & \ldots & 0 \\
\ldots \\
0 & 0 & \ldots & 1 \\
b_1 & b_2 & \ldots & b_{n+1}
\end{array}\right)
\end{array}$$


Now, of course, (\ref{matrixeq}) has solution
$$X_k=A^{k-1} X_1 + \sum_{i=1}^k A^{k-i} B E_{i}, $$
where $X_1$ is the vector of initial conditions of this linear
dynamical system.
If $A$ has eigenvalues (roots of $x^n-\sum_{i=1}^{n}a_{i} x^{i-1}$)
of module strictly less than 1, then the term $A^{k-1} X_1$ will tend to
zero when $k$ tends toward infinity, whereas the partial sums
$\sum_{i=1}^k A^{k-i} B E_{i}$ will tend towards a finite value
(obtained as a convergent infinite series). 

\begin{exem}
\label{transf2}
Consider the following filter of order 2 (see \cite{SAS07}):
$$
x_i = 0.7 e_i - 1.3 e_{i-1} + 1.1 e_{i-2} + 1.4 x_{i-1} - 0.7 x_{i-2},
$$
where $e_i$ are independent inputs between 0 and 1. A typical run
of this algorithm with $e_i=\frac{1}{2}$ (for all $i$)
and $x_0=0$, $x_1=0$ converges
towards 0.8333... stays positive, and reaches at most 1.1649 
(its dynamics is shown in Figure \ref{dynamics1}).
\begin{center}
\begin{figure}
\begin{minipage}{4cm}
\begin{center}
\epsfig{file=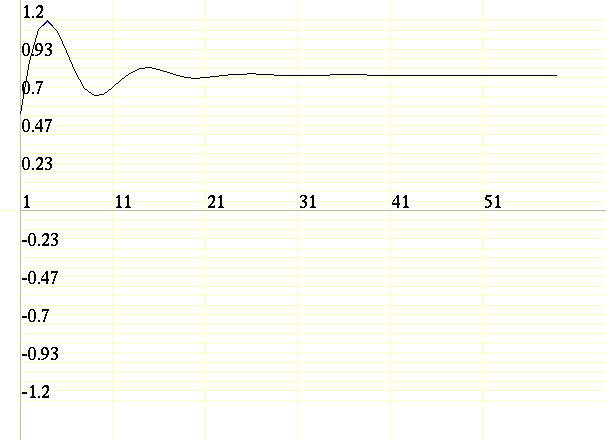,width=4cm,clip=}
\caption{A run of the filter example.}
\label{dynamics1}
\end{center}
\end{minipage}
\begin{minipage}{4cm}
\begin{center}
\epsfig{file=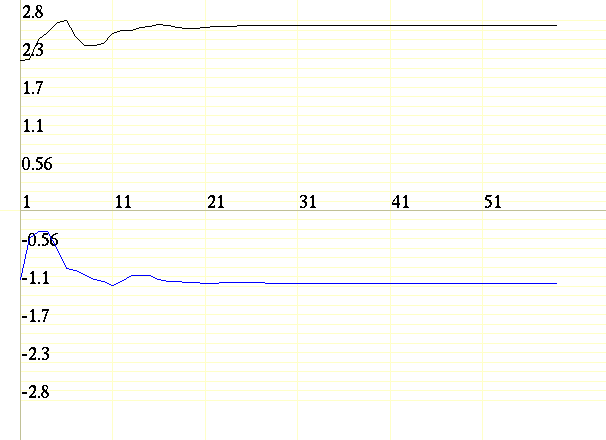,width=4cm,clip=}
\caption{Min and max over the iterations.}
\label{dynamics2}
\end{center}
\end{minipage}
\end{figure}
\end{center}
\end{exem}

\subsection{Classical affine arithmetic}

\label{affine}

An {\em affine form}
is a formal series over a set of {\em noise symbols} $\varepsilon_i$
\[ \hat x = \alpha^x_0 + \sum_{i=1}^{\infty} \alpha^x_i \varepsilon_i,\]
with 
$\alpha^x_i \in \R$.

Let $\RAff$ denote the set of such affine forms. 
Each noise symbol $\varepsilon_i$ stands for an independent component of the total
uncertainty on the quantity $\hat x$, its value is unknown but bounded in [-1,1]; 
the corresponding coefficient $\alpha^x_i$ is a known real value, which gives 
the magnitude of that component. The idea is that the same noise symbol can be 
shared by several quantities, indicating correlations among them. These noise 
symbols can be used not only for modeling uncertainty in data or parameters,
but also uncertainty coming from computation.

When the cardinal of the set $\{\alpha^x_i \neq 0\}$ is finite, such affine
forms correspond to the affine forms introduced first in \cite{Stolfi} and
defined for static analysis in \cite{SAS06} by the authors. 

The concretization of a set of affine forms sharing noise symbols is
a center-symmetric polytope, which center is given by the $\alpha_0$ vector 
of the affine forms. 
For example, the concretization of
\begin{eqnarray*}
 \hat x &=& 20-4 \varepsilon_1 +2 \varepsilon_3 + 3 \varepsilon_4 \\
 \hat y &=& 10-2\varepsilon_1 + \varepsilon_2 - \varepsilon_4
\end{eqnarray*}
is given in Figure \ref{fig_concret}.
\begin{figure}
\begin{center}
\epsfig{file=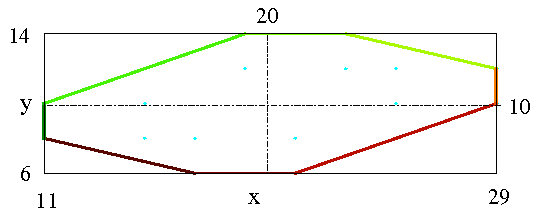,width=8cm,clip=}
\caption{Concretization is a center-symmetric polytope}
\label{fig_concret}
\end{center}
\end{figure}

$\RAff$ is a $\R$-vector space with the operations $+$ and $\times$:
$$\begin{array}{ll}
\left( \alpha^x_0 + \sum_{i=1}^{\infty} \alpha^x_i \varepsilon_i
\right) + & \left( \alpha^y_0 + \sum_{i=1}^{\infty} \alpha^y_i \varepsilon_i
\right)= \\
& (\alpha^x_0+\alpha^y_0)+\sum_{i=1}^{\infty} (\alpha^x_i+\alpha^y_i) \varepsilon_i
\end{array}$$
$$\begin{array}{ll}
\lambda \times \left( \alpha^x_0 + \sum_{i=1}^{\infty} \alpha^x_i \varepsilon_i
\right) = & \lambda \alpha^x_0 + \sum_{i=1}^{\infty} \lambda \alpha^x_i \varepsilon_i
\end{array}
$$

A sub-vector space $\RAff_1$ of $\RAff$ 
can be classically endowed with a Banach space structure, using
the ${\ell}_1$ norm
$$\norm{\hat{x}}_1=\sum_{i=0}^{\infty} |\alpha^x_i|,
$$
for the elements $\hat{x}$ such that the above sum is finite.

We define the projection 
\[ \pi_L(\hat{x}) = \sum_{i=1}^{\infty} \alpha^x_i \varepsilon_i, \]
and the associated semi-norm
\[ \normL{\hat{x}} = \norm{\pi_L(\hat{x})}_1. \]
We also define $\normA{\hat{x}} = |\alpha_0^x| + \normL{\hat{x}}.$

\vskip .2cm

Let us come back now to the linear recursive filters of Section \ref{numschemes}.
For simplicity's sake, suppose $e_{k+n+1}=\varepsilon_{k+n+1}$, so the $e_{k+n+1}$ are
independent inputs between 0 and 1. As we will see later on, but as
should already be obvious because of the definitions of sum and product
of affine forms by a scalar, the semantics, using affine forms, of the completely unfolded
filter program exactly gives, at unfolding $k$:
$$\hat{X}_k=A^{k-1} \hat{X}_1 + \sum_{i=1}^k A^{k-i} B \hat{E}_i$$
where $$\hat{E}_i=\left(\begin{array}{c}
\varepsilon_{i-n} \\
\varepsilon_{i-n+1} \\
\ldots \\
\varepsilon_i 
\end{array}\right)$$ and $\hat{X}_i$ are the obvious affine forms vector
counterparts of $X_i$.

This means that in the case $A$ has complex roots of module strictly
less than 1, the affine forms (with a finite number of $\varepsilon_i$)
giving the semantics of values at unfolding $k$, converge in the
$\ell_1$ sense to an affine form with infinitely many noise symbols.

\begin{exem}
Consider again the filter of order 2 of Example \ref{transf2}.
We supposed that the successive inputs $e_i$ 
are independent inputs between 0 and 1, so that 
we can write $\hat{e}_i = \frac{1}{2}+\frac{1}{2} \varepsilon_{i+1}$
(with different noise symbols at each iterate), and $x_0= x_1 = 0$. 
For instance, if we compute the affine form after 99 unfolds, we
find:
$$\begin{array}{rcl}
\hat{x}_{99} & = & 0.83 +7.81 e^{-9} \varepsilon_1 
- 2.1 e^{-8} \varepsilon_2
 - 1.58 e^{-8} \varepsilon_3 +\ldots\\
& & -0.16 \varepsilon_{99}+0.35 \varepsilon_{100}
\end{array}$$
 whose concretization gives an exact (under the assumption that the coefficients of the affine form 
are computed with arbitrary precision) enclosure of $x_{99}$~:
$$x_{99}
\sqsubseteq
[-1.0907188500,2.7573854753].$$
The limit affine form has a concretization converging towards (see
Figure \ref{dynamics2}):
$x_{\infty} = [-1.09071884989...,2.75738551656...]$.
\end{exem}

Unfortunately, if asymptotically (i.e. when $k$ is large enough), 
the concretization of the affine
forms $\hat{X}_k$ converges to a good estimate of the values that
program variable $x$ can take (meaning, after a large number
of iteration $k$), this form is in no way an invariant of the loop, and
{\em does not account} for all values that this variable can take along
the loops.

\begin{exem}
\label{transf22}
This can be seen for the particular filter of order 2 of Example
\ref{transf2}. In Figure \ref{dynamics2}, the reader with good eyes can
spot that around iterations 8-10, the concretization of $x$ can
go above 2.75738551656..., which is the {\em asymptotic} supremum. 
Actually, the sup value is 2.82431841..., reached at iteration 8,
whereas the infimum is -1.12124069..., reached at iteration 13.
\end{exem}

{\bf The aim of this paper is to describe a suitable extension of these
affine forms that can account for such invariants.}

\section{An ordered Banach space of generalized affine forms}

\label{poBanach}


We now extend our Banach space of affine forms in order to {\em represent}
unions of affine forms, as a {\em perturbed} affine form. We
consider $\RA_1=\RAff_1 \oplus \R$ and write these new affine forms
as:
\[ \hat x = \alpha^x_0 + \sum_{i=1}^{\infty} \alpha^x_i \varepsilon_i
+\beta^x \varepsilon_U\]
Norms $\ell_1$ are extended over this new domain in an obvious manner.
We now have $\norm{\hat{x}}_1 =  |\alpha_0^x| + \normL{\hat{x}} + |\beta^x|.$

\paragraph{Remark: }
In the rest of this section, unless otherwise stated, 
we restrict the study to elements in 
the {\em cone} $\RA_+$ of $\RA_1$ whose elements $\hat{x}$ have a positive $\beta^x$.
We will sketch some ways to extend the results obtained, and their meaning, for
all of $\RA_1$, in Sections \ref{meetoperation} and \ref{betaneg}.

\vskip .2cm

We first give concrete semantics to these generalized affine forms in 
Section \ref{concsem}, then we give in Section \ref{arithexpr} the abstract transfer functions
for arithmetic expressions. 
The counterpart of the inclusion ordering, the {\em continuation}
ordering, is defined in Section \ref{continuation}. The main
technical ingredient that will allow us to find effective join and
meet operations in Section \ref{quasilattice} is the equivalence between
this seemingly intractable ordering, and an ordering with a much simpler
definition,
the {\em perturbation} ordering, see Theorem \ref{maintheo1}.
Finally we prove in Section \ref{Banach} that these
generalized affine forms also have the structure
of {\em an ordered} Banach space. This will be useful for proving convergence
results with our iteration schemes in Section \ref{linear}.

\subsection{Concrete semantics of expressions and concretization function}
\label{concsem}

\begin{defi}
We define the concretization function $\gamma: \RA_+ \rightarrow \I$ in intervals as
follows, for $\hat{x}=\alpha^x_0+\sum_{i=1}^{\infty} \alpha^x_i \varepsilon_i + \beta^x \varepsilon_U$:
$$\gamma\left(\hat{x}\right)=
[\alpha^x_0 -\normL{\hat{x}}-\beta^x, 
\alpha^x_0 + \normL{\hat{x}} + \beta^x]
$$
\end{defi}
whose lower (respectively upper) bound corresponds to the infimum
(respectively supremum) of the affine form $\hat{x}$ seen as a function
from $\varepsilon_i \in [-1,1]$ to $\R$. 

Let $\var$ be the set of program variables. An {\em abstract environment}
is a function $\sigma: \var \rightarrow \RA_+$. We write $\hat{\var}$
for the set of such abstract environments. The fact that the affine
forms representing the variables share some common noise symbols can
be expressed in the {\em joint concretization}, also denoted by $\gamma$,
of $\sigma$ (we suppose here that $\var$ is finite and equal
to $\{x_1,\ldots,x_k\}$):
$$\gamma(\sigma)=\left\{
\begin{array}{l}
(\hat x_1,\ldots,\hat x_k) \ \mid \ \exists t_1,\ldots,
t_n, \ldots \in [-1,1], \\ \exists u_{x_1},\ldots,u_{x_k} \in [-1,1], \\ 
 \left\{\begin{array}{rcl}
\hat x_1 & = & \alpha^{x_1}_0 + \sum_{i=1}^{\infty} \alpha^{x_1}_i t_i + \beta^{x_1} u_{x_1} \\
\ldots \\
\hat x_k & = & \alpha^{x_k}_0 + \sum_{i=1}^{\infty} \alpha^{x_k}_i t_i + \beta^{x_k} u_{x_k} \\
\end{array}\right.\end{array}\right\}$$

\begin{exem}
\label{firstjointaffine}
Consider 
$$\begin{array}{rcl}
\hat{x} & = & 1+\varepsilon_1+\varepsilon_2+\varepsilon_U \\
\hat{y} & = & 2-\varepsilon_1+2\varepsilon_2+\epsilon_U
\end{array}$$
Their joint concretization is the inner polyhedron of Figure \ref{inclusionaffine}.
\end{exem}

Seeing affine forms $\hat{x}$ as functions of $\varepsilon_i$ and $\varepsilon_U$, we
define the {\em concrete} semantics of arithmetic operations $+$, $-$ and $\times$
on affine forms, with values in the set of subsets of $\R^k$, as follows. 
We note $\hat{x}(t_1,\ldots,t_n,\ldots,u_x)$, for
$t_1,\ldots,t_n,\ldots,u_x \in [-1,1]$, the application $\hat{x}$, seen as an affine 
function
of $\varepsilon_1,\ldots,\varepsilon_n,\ldots, \varepsilon_u$,
to $t_1,\ldots,t_n,\ldots,u_x$.

The concrete semantics of $\hat{x}+\hat{y}$ in $\R$ is now

\noindent
$\begin{array}{rcl}
\im \hat{x}+\hat{y} & = & \left\{\hat{x}(t_1,\ldots,t_n,\ldots,u_x)+\hat{y}(t_1,\ldots,t_n,\ldots,u_y) \right.\\
& & \left.\mid t_1,\ldots,t_n,\ldots,u_x,u_y \in [-1,1]\right\}
\end{array}$

For $-\hat{x}$, it is

\noindent
$\begin{array}{rcl}
\im -\hat{x} & = & \left\{-\hat{x}(t_1,\ldots,t_n,\ldots,u_x) \right.\\
& & \left.\mid t_1,\ldots,t_n,\ldots,u_x \in
[-1,1]\right\} 
\end{array}$

Finally, the concrete semantics of $\hat{x}\times \hat{y}$ is

\noindent
$\begin{array}{rcl}
\im \hat{x}\times\hat{y} & = & \left\{\hat{x}(t_1,\ldots,t_n,\ldots,u_x)\times \hat{y}(t_1,\ldots,t_n,\ldots,u_y) \right.\\
& & \left.\mid t_1,\ldots,t_n,\ldots,u_x,u_y \in [-1,1]\right\} 
\end{array}$

We are going to give an abstract semantics for $+$, $-$ and $\times$ in 
next section.

\subsection{Abstract interpretation of simple arithmetic expressions}

\label{arithexpr}

Let $Expr$ be the set of polynomial expressions\footnote{Nothing prevents
us from defining abstract transfer functions for other operations, such
as $\sqrt{.}$, $sin$, $acos$ etc. as affine forms are naturally
Taylor forms. This is not described in this article, for lack of space.}, i.e. expressions 
built inductively from the set of program variables $\var$, real number
constants, 
and operations +, - and $\times$.
We now define the respective operators $\hat{+}$, $\hat{-}$ and
$\hat{\times}$ (extending the ones of Section \ref{affine})

\begin{defi}
$$\begin{array}{rcl}
\hat{x}\hat{+}\hat{y} & = & \alpha^x_0+\alpha^y_0 + 
\sum_{i=1}^{\infty}\left(\alpha^x_i+\alpha^y_i\right)\varepsilon_i 
+\left(\beta^x+\beta^y\right) \varepsilon_{U} \\
\hat{-}\hat{x} & = & 
-\alpha^x_0 - \sum_{i=1}^{\infty} \alpha^x_i \varepsilon_i + \beta^x \varepsilon_{U}
\end{array}
$$
(note that the sign $+$ in $+\beta^y \varepsilon_U$ is certainly not a typo).
And we define\footnote{Better abstractions are available, but make the
presentation more complex, this is left for the full version of this article.} for affine forms $\hat{x}$ and $\hat{y}$ having a finite
number of non-zero $\alpha_i$ coefficients (we call them affine forms with finite support)
\begin{eqnarray*}
\hat{x}\hat{\times}\hat{y}& = &
\alpha^x_0 \alpha^y_0 + 
\sum_{i=1}^{\infty} \left(\alpha^x_0 \alpha^y_i + \alpha^x_i \alpha^y_0\right) 
\varepsilon_i +\sum_{i,k=1}^{\infty} \mid \alpha^x_i \alpha^y_k \mid
 \varepsilon_{f} \\
&&+ \sum_{j=0}^{\infty} \left(\mid \alpha^x_j \mid \beta^y 
+\beta^x \mid \alpha^y_j \mid + \beta^x \beta^y
\right)\varepsilon_{U},
\end{eqnarray*}
where $\varepsilon_f$ is a symbol which is unused in $\hat{x}$
nor in $\hat{y}$ (``fresh
noise symbol'').
\end{defi}

\begin{lemm}
\label{base}
We have the following correctness result on the abstract semantics of expressions:
$$\begin{array}{rclcl}
\im \hat{x}+\hat{y} & \subseteq & \gamma(\hat{x} \hat{+} \hat{y}) 
& \subseteq & \gamma(\hat{x})+\gamma(\hat{y}) \\
\im -\hat{x} & \subseteq & \gamma(\hat{-}\hat{x}) & \subseteq &
-\gamma(\hat{y}) \\
\im \hat{x}\times \hat{y} & \subseteq & \gamma(\hat{x}\hat{\times} \hat{y})
\end{array}
$$
where $+$ and $-$ on the right hand side of inequalities above are the
corresponding operations on intervals, and the last inclusion holds only for 
affine forms with finite support.
\end{lemm}

\begin{sketch}
This is mostly the similar classical result in affine arithmetic \cite{Stolfi}, and easily
extended to $\varepsilon_U$ symbols and infinite series (convergent in the $\ell_1$ sense).
$\Box$
\end{sketch}

\subsection{The continuation and the perturbation ordering}

\label{continuation}


The correctness of the semantics of arithmetic expressions defined in 
Section \ref{arithexpr}, and more generally of the semantics of a
real language (Section \ref{imp}) relies on an information 
ordering, which we call the {\em continuation ordering}, Definition
\ref{continuationorder}. Unfortunately,
its definition makes it difficult to use, and we define an a priori weaker 
ordering, that we call {\em perturbation ordering}, Definition and Lemma
\ref{perturborder}, that will be
easily decidable, and shown equivalent to the continuation ordering 
(Proposition \ref{equivinfoperturb}). The {\em perturbation ordering} has
minimal upper bounds, but not least upper bounds. A simple construction
will allow us to define in Section \ref{riesz} a lattice with a
slightly stronger {\em computational} ordering, based on the
{\em perturbation ordering}. 

\begin{defi} [continuation order]
\label{continuationorder}
Let $\sigma_1$ and $\sigma_2$ be two abstract environments. We say that
$\sigma_1 \preceq \sigma_2$ if and only if for all $e \in Expr$
$$\gamma \semb e \seme \sigma_1 \subseteq 
\gamma \semb e \seme \sigma_2$$

We naturally say that
$\hat{x} \preceq \hat{y}$ if and only if 
$$\gamma \semb e \seme \sigma[u \leftarrow \hat{x}] \subseteq 
\gamma \semb e \seme \sigma[u \leftarrow \hat{y}]$$
for all $e \in Expr$ and all $\sigma\in \hat{Var}$ (and for some $u \in \var$).
\end{defi}

\begin{deflem} [perturbation order]
\label{perturborder}
We define the following binary relation $\leq$ on elements of $\RA_1$
$$x \leq y \Leftrightarrow \normA{x-y}\leq \beta^y - \beta^x.$$
Then $\leq$ is a partial order on $\RA_1$.

We extend this partial order {\em componentwise} to abstract environments
as follows: for all $\sigma_1, \sigma_2: \var \rightarrow \RA_1$, 
$$\sigma_1 \leq \sigma_2 \Leftrightarrow \forall x \in \var, \sigma_1(x) \leq
\sigma_2(x)$$

\end{deflem}

\begin{sketch}
Reflexivity and transitivity of $\leq$ are trivial. For antisymmetry, 
suppose $\hat{x}\leq \hat{y}$ and $\hat{y} \leq \hat{x}$, then we have
$$\begin{array}{rcl}
\normA{\hat{x}-\hat{y}} & \leq & \beta^y - \beta^x \\
\normA{\hat{x}-\hat{y}} & \leq & \beta^x - \beta^y 
\end{array}$$
This implies that both $\beta^y-\beta^x$ and $\beta^x - \beta^y$ are
positive, hence necessarily zero. Hence also $\normA{\hat{x}-\hat{y}}=0$ meaning
$\pi_A{\hat{x}}=\pi_A{\hat{y}}$. Overall: $\hat{x}=\hat{y}$. 
$\Box$
\end{sketch}

Now, we prove intermediary results in order to prove equivalence
between the two orders above. Half of this equivalence is easy, see
Lemma \ref{predleq}. The other half is a consequence of Lemma \ref{conc}
and of Lemma \ref{increasing}. Theorem \ref{maintheo1} is the same as 
Proposition \ref{equivinfoperturb}, 
not just for individual affine forms, but for
all abstract environments.

\begin{lemm}
\label{predleq}
$\hat{x} \preceq \hat{y} \Rightarrow \hat{x} \leq \hat{y}$
\end{lemm}

\begin{proof}
Given $\hat{x}$ and $\hat{y}$ in $\RA_1$, consider the expression $e=u-v$ and
the environment $\sigma$ such that $\sigma(v)=\pi_L(\hat{y})+\alpha_0^y$. 
We have $\gamma \semb e \seme [u\leftarrow \hat{x}] 
\subseteq \gamma \semb e
\seme \sigma [u \leftarrow \hat{y}]$, which means:
$$
\left[\alpha^x_0-\alpha^y_0-\normL{x-y}-\beta^x,
\alpha^x_0-\alpha^y_0+\normL{x-y}+\beta^x\right]$$
$$
\subseteq 
[-\beta^y,\beta^y]
$$
Thus we have 
\begin{eqnarray}\label{first}
\alpha^x_0-\alpha^y_0 & \leq & -\normL{x-y}+\beta^y-\beta^x, \\ 
\alpha^x_0-\alpha^y_0 & \geq & \normL{x-y}+\beta^x-\beta^y. \label{last}
\end{eqnarray}
Inequality (\ref{last}) is equivalent to
$$\begin{array}{rcl}
\alpha^y_0-\alpha^x_0 & \leq & -\normL{x-y}+\beta^y-\beta^x, \\
\end{array}$$
hence together with  inequality (\ref{first})
\[ |\alpha^x_0-\alpha^y_0| \leq -\normL{x-y}+\beta^y-\beta^x,  \]
this exactly translates into $\hat x \leq \hat y $.
$\Box$
\end{proof}

\begin{lemm}
\label{conc}
For all $\hat{x}, \hat{y} \in \RA_1$, $\hat{x} \leq \hat{y}$ implies $\gamma(\hat{x}) \subseteq
\gamma(\hat{y})$.
\end{lemm}

\begin{proof}
We compute:
$$
\sup \/ \gamma(\hat{y}) - \sup \/ \gamma(\hat{x}) = \alpha_0^y - \alpha_0^x
+ \normL{\hat{y}}-\normL{\hat{x}} 
 +\beta^y-\beta^x $$
Using the triangular inequality $\normL{\hat{x}} \leq \normL{\hat{x}-\hat{y}}+\normL{\hat{y}}$, and 
 $\normA{x-y} \leq \beta^y - \beta^x$, we write:
$$\begin{array}{rcl}
\sup \/ \gamma(\hat{y}) - \sup \/ \gamma(\hat{x}) & \geq & \alpha_0^y - \alpha_0^x
- \normL{\hat{x}-\hat{y}} +\beta^y-\beta^x \\
& \geq & \alpha_0^y - \alpha_0^x +|\alpha_0^y - \alpha_0^x| \geq 0
\end{array}$$
and similarly for the $\inf$ bound of
the concretization.
$\Box$
\end{proof}

Notice that the converse of Lemma \ref{conc} is certainly not true: just
take $\hat{x} = 1 + \varepsilon_1$ and $\hat{x'} = 1 + \varepsilon_2$. 
It is easy to see that $\hat{x}$ and $\hat{x'}$ are incomparable, but
have same concretizations.

\begin{lemm}
\label{increasing}
$\hat{+}$, $\hat{-}$ and $\hat{\times}$ are increasing functions
on $(\RA_+,\leq)$.
\end{lemm}

\begin{proof}
We have easily, for $\hat{x} \leq \hat{y}$ and $\hat{z} \in \RA_+$:
$$\begin{array}{lll}
\normA{\hat{x}\hat{+}\hat{z}-\hat{y}\hat{+}\hat{z}} &=& \normA{\hat{x}-
\hat{y}} 
\leq  \beta^y - \beta^x =  \beta^{y+z} - \beta^{x+z} \\
\normA{\hat{-}\hat{x}-\hat{-}\hat{y}} &=& \normA{\hat{x}-\hat{y}}
\leq \beta^y - \beta^x = \beta^{-y} - \beta^{-x}
\end{array}$$
Now:
\begin{eqnarray}
\normA{\hat{x}\hat{\times}\hat{z}-\hat{y}\hat{\times}\hat{z}} & = &
\normA{z}\normA{x-y}\\
& \leq & \normA{z}\left(\beta^y - \beta^x \right)
\label{eq1}
\end{eqnarray}
$$\begin{array}{rcl}
\beta^{y \times z} - \beta^{x \times z} & = &
\normA{z}(\beta^y - \beta^x)+\\
& & \beta^z (\normA{\hat{y}}-\normA{\hat{x}}
+\beta^y - \beta^x )
\end{array}$$
But $\hat{x}\leq \hat{y}$ so 
$
\beta^y-\beta^x \geq \normA{\hat{y}-\hat{x}} \geq \normA{\hat{x}}-\normA{\hat{y}}$
the last inequality being entailed by the triangular inequality.
Thus, 
$$\begin{array}{rcl}
\beta^{y \times z} - \beta^{x \times z} & \geq & \normA{z}(\beta^y - \beta^x)
\end{array}$$
which, by combining with inequality (\ref{eq1}), completes the proof.
$\Box$

\end{proof}

\begin{prop}
\label{equivinfoperturb}
$\hat{x} \leq \hat{y}$ if and only if $\hat{x} \preceq \hat{y}$
\end{prop}

\begin{sketch}
We know from Lemma \ref{predleq} that $\hat{x} \preceq \hat{y}$ implies
$\hat{x} \leq \hat{y}$. Now, let $e \in Expr$, and suppose $\hat{x}
\leq \hat{y}$. We reason by induction on $e$: the base case is constants
and variables (trivial). 

A consequence of Lemma \ref{increasing} is
that for all $\hat{z} \in \RA_+$,
$\hat{x}\hat{+}\hat{z} \leq \hat{y}\hat{+}\hat{z}$,
$\hat{-}\hat{x}\leq \hat{-}\hat{y}$ and 
$\hat{x}\hat{\times}\hat{z} \leq \hat{y}\hat{\times}\hat{z}$.
By induction on the syntax of $e$,
we then have $\semb e \seme \sigma[u \leftarrow \hat{x}] 
\leq \semb e \seme \sigma[u \leftarrow \hat{y}]$. 
This implies by Lemma \ref{conc} that 
$\gamma \semb e \seme \sigma[u \leftarrow \hat{x}] \subseteq 
\gamma \semb e \seme \sigma[u \leftarrow \hat{y}]$, hence
$\hat{x} \preceq \hat{y}$.
$\Box$
\end{sketch}

\vskip .2cm

Finally, we can prove the following more general equivalence, which is nothing but
obvious at first. The example below shows the subtlety of this result.

\begin{exem}
\label{secondjointaffine}
To illustrate one of the aspects of next theorem, that is, 
$\hat{x} \leq \hat{x'}$ implies that any {\em joint concretization} of
$\hat{x'}$ with other affine forms, (say just one, $\hat{y}$, here), contains
the joint concretization of $\hat{x}$ with $\hat{y}$, take again 
$\hat{x}$ as in Example \ref{firstjointaffine} and
$$\begin{array}{rcl}
\hat{x'} & = & \frac{3}{2}+\frac{1}{2}\varepsilon_1+\varepsilon_2+2\varepsilon_U
\end{array}$$
Of course, $\hat{x}\leq \hat{x}'$; 
Figure \ref{inclusionaffine} shows the inclusion of the joint concretization
of $(\hat{x},\hat{y})$ in the joint concretization of $(\hat{x'},\hat{y})$.
Note that several of the faces produced are fairly different.
\end{exem}

\begin{center}
\begin{figure}
\begin{center}
\epsfig{file=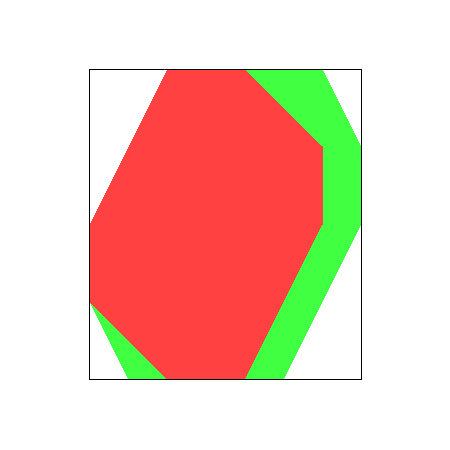,width=6cm, height=3cm, clip=}
\caption{Joint concretization of Example \ref{firstjointaffine}, included in the joint concretization of
Example \ref{secondjointaffine}.}
\label{inclusionaffine}
\end{center}
\end{figure}
\end{center}

\begin{theo}
\label{maintheo1}
Let $\sigma_1$, $\sigma_2$ be two abstract environments, then 
 $\sigma_1 \leq \sigma_2$ if and only if $\sigma_1 \preceq \sigma_2$
\end{theo}

\begin{sketch}
It can be shown first (classical result in affine arithmetic \cite{Stolfi}),
that $\gamma(\sigma_2)$ is a polyhedron (a particular kind, called a
{\em zonotope}). It means that it can be equivalently described by a
system of affine constraints ($j=1,\ldots,k$):
$$\sum_{x \in \var} a^j_x x \leq b^j$$

Consider the expressions (in $Expr$): $e^j=\sum_{x \in \var} a^j_x x - b^j$.
We know that for all $x \in \var$, $\sigma_1(x) \leq \sigma_2(x)$, hence
by Proposition \ref{equivinfoperturb}, $\sigma_1(x) \preceq \sigma_2(x)$.
This entails, by induction on $\var$, that $\gamma \semb e^j \seme \sigma_1
\subseteq \gamma \semb e^j \seme \sigma_2$. Thus the constraint
$\sum_{x \in \var} a^j_x x \leq b^j$ is satisfied by elements of
$\gamma(\sigma_1)$, by Lemma \ref{base}. So $\gamma(\sigma_1) \subseteq 
\gamma(\sigma_2)$.

Let $e$ be any expression in 
$Expr$. The result follows from Proposition \ref{equivinfoperturb}
and the result above, by induction on $e$.
\end{sketch}

\subsection{Ordered Banach structure}

\label{Banach}

The aim of this section is to prove Proposition \ref{orderedBanach}.
This will be central to the proofs in Sections \ref{riesz} and
\ref{convsection}. 

\begin{prop}
\label{orderedBanach}
$(\RA_1,\leq)$ is an ordered Banach space.
\end{prop}

\begin{sketch}
First, we show that the partial order $\leq$ of Definition and 
Lemma \ref{perturborder} makes $\RA_1$ into an ordered
vector space. 

For showing this, we have
to show {\em compatibility} of $\leq$ with the linear structure, i.e., for
$\lambda \geq 0$ and $\hat{x} \leq \hat{y}$, and for all $\hat{z}$:
$\hat{x}+\hat{z} \leq \hat{y}+\hat{z}$, 
$\lambda \hat{x} \leq \lambda \hat{y}$, and
$-\hat{y} \leq -\hat{x}$, 
which is immediate verification.

The only remaining property to prove is that
$\leq$ is closed in $\RA_1 \times \RA_1$, in the product topology, 
$\RA_1$ being given the Banach topology of the $\ell_1$ norm.
Suppose $x_n$ converges towards $x$ as $n$ goes towards $\infty$, in the sense of the $\ell_1$ norm, and
suppose for all $n$, $x_n \leq y$. Then
$$\begin{array}{rcl}
\normA{y-x} & \leq & \normA{y-x_n}+\normA{x_n-x} \\
& \leq & \pi_U(y)-\pi_U(x_n)+\frac{\varepsilon}{2}
\end{array}$$
for all $\epsilon > 0$ and $n \geq N(\epsilon)$. By continuity of
$\pi_U$, we thus know that there exists $K(\epsilon)$ such that
for all $n \geq \sup(N(\epsilon),K(\epsilon))$:
$$\begin{array}{rcl}
\normA{y-x} & \leq & \normA{y-x_n}+\normA{x_n-x} \\
& \leq & \pi_U(y)-\pi_U(x)+\epsilon
\end{array}$$
This concludes the proof. 
$\Box$
\end{sketch}

A different way of stating that $(\RA_1,\leq)$ is an ordered vector
space is to introduce the subset ${\cal C}$ 
of $\RA_1$ such that
$$x \leq y \Leftrightarrow y-x \in {\cal C},$$
and show it is indeed the {\em cone} of $\leq$, see \cite{conesduality}.
We see that $${\cal C}=\left\{x \mid \normA{x} \leq \beta^x\right\}.$$
This is the analogue of the {\em Lorentz cone} in special relativity
theory, but with the $\ell_1$ norm instead of the $\ell_2$ norm.

To use the vocabulary from relativity theory, identifying the $\RAff$ part
with the space coordinates and the $\beta$ coefficient with the time
coordinate, $\leq$ is the {\em causal order} and
$\hat{x}\leq \hat{y}$ if the {\em space-time interval}
$[\hat{x},\hat{y}]$ is {\em time-like} or {\em light-like}, 
whereas 
$\hat{x} \geq \hat{y}$ if $[x,y]$ is {\em space-like} or {light-like}. 
Other considerations, using domain-theoretic methods, on the causal order in the
case of the $\ell_2$ Lorentz cone can be found for instance in
\cite{Panangaden}.

\subsection{The quasi lattice structure}

\label{quasilattice}

We will show in this section that $(\RA_1,\leq)$ is almost a bounded
complete partial order (bcpo). 
It is not a bcpo because there is not in general any
least upper bound. This is a consequence of \cite{Krein48}: as the cone $C$ of
our partial order has $2^n$ generators (the generators of the polyhedron
which is the unit $\ell 1$ ball), it cannot be simplicial, hence $(\RA_1,\leq)$
is not a lattice. Instead, there are in general infinitely many
{\em minimal upper bounds}, which will suffice for our semantics purposes. 
We prove furthermore that {\em many} bounded subset of
$\RA_1$ (``enough'' again) admit minimal upper bounds. 

We first recall the definition of a minimal upper bound or {\em mub}
(maximal lower bounds, or {\em mlb}, are defined similarly):

\begin{defi}
Let $\sqsubseteq$ be a partial order on a set $X$. We say that
$z$ is {\em a mub} of two elements $x,y$ of $X$ if and only if
\begin{itemize}
\item $z$ is an upper bound of $x$ and $y$, i.e. $x \sqsubseteq z$ and
$y \sqsubseteq z$,
\item for all $z'$ upper bound of $x$ and $y$,
$z'\sqsubseteq z$ implies $z=z'$.
\end{itemize}
\end{defi}

We note that for the order $\leq$, we have a very simple characterization
of mubs, if they exist (proving existence, and deriving some formulas, when available, 
are the aims of the section to come).

\begin{lemm}
\label{mubs}
\label{necmin}
Let $\hat{x}$ and $\hat{y}$ be two elements of $\RA_1$. Then $\hat{z}$ is
a mub of $\hat{x}$ and $\hat{y}$ if and only if
\begin{itemize}
\item $\hat{x} \leq \hat{z}$ and $\hat{y} \leq \hat{z}$,
\item $\beta^z$ is minimal among the $\beta^t$, for all $\hat{t}$ upper bounds
of $\hat{x}$ and $\hat{y}$. 
\end{itemize}
\end{lemm}

\begin{proof}
Suppose we have $\hat{z}$ such as defined above. Take any upper bound
$\hat{t}$ of $\hat{x}$ and $\hat{y}$ and suppose $\hat{t}\leq \hat{z}$.
Then:
$\normA{\hat{z}-\hat{t}} \leq \beta^z - \beta^t$.
Hence, $\beta^z \geq \beta^t$. But by hypothesis, $\beta^z$ is minimal
among all upper bounds, so $\beta^z=\beta^t$. Then this implies
$\normA{\hat{z}-\hat{t}}=0$ so $\pi_A(\hat{z})=\pi_A(\hat{t})$ as well,
hence $\hat{z}=\hat{t}$.
$\Box$
\end{proof}

In what follows, we will need an extra definition:

\begin{defi}
Let $x$ and $y$ be two intervals. We say that $x$ and $y$ are in 
generic positions if, whenever $x \subseteq y$, $\inf x = \inf y$ or
$\sup x = \sup y$. 

By extension, we say that two affine forms $\hat{x}$ and $\hat{y}$ are
in generic position when $\gamma(\hat{x})$ and $\gamma(\hat{y})$ are
intervals in generic positions.
\end{defi}

\subsubsection{The join operation}
\label{sec_join}
For any interval $i$, we note $mid(i)$ its center. Let $\alpha_i^x \wedge \alpha_i^y$ denote the minimum of the 
two real numbers, and $\alpha_i^x \vee \alpha_i^y$ their maximum. We define
\[
\argmin{\alpha^x_i}{\alpha^y_i} = \{\alpha \in [\alpha_i^x \wedge \alpha_i^y,\alpha_i^x \vee \alpha_i^y],
 |\alpha| \mbox{ minimal}\}
\]
\[
\argmax{\alpha^x_i}{\alpha^y_i} = \{\alpha \in [\alpha_i^x \wedge \alpha_i^y,\alpha_i^x \vee \alpha_i^y], \\
 |\alpha| \mbox{ maximal}\}
\]

\begin{prop}
\label{join1}
Let $\hat x, \hat y \in \RA_1$. There exist minimal upper bounds $\hat z$ of $\hat x$ and $\hat y$ if and only if 
\begin{equation}
\label{exist_mub}
\normA{\hat x- \hat y} \geq | \beta^y - \beta^x|
\end{equation}
Moreover, the minimal upper bounds, when they exist, all satisfy 
\begin{equation} \beta^z = \frac12 (\normA{\hat x- \hat y} + \beta^x + \beta^y),
\label{beta_opt}
\end{equation}
\begin{equation}
\alpha_i^x \wedge \alpha_i^y \leq  \alpha_i^z \leq \alpha_i^x \vee \alpha_i^y, \; \forall i \geq 0,
\label{minmax}
\end{equation}
and they are such that 
$\normA{\hat x- \hat z} = \beta^z - \beta^x$ and $\normA{\hat y-\hat z} = \beta^z - \beta^y$.

\end{prop}

\begin{proof}
We first characterize $\beta^z$ by expressing $\hat x \leq \hat z$ and $\hat y \leq \hat z$:
\[ \normA{\hat x-\hat y} \leq \normA{\hat x-\hat z} + \normA{\hat y-\hat z} \leq \beta^z - \beta^x + \beta^z - \beta^y.\] 
The smallest possible $\beta^z$ thus is 
$\beta^z = \frac12 (\normA{\hat x-\hat y} + \beta^x + \beta^y).$ 
Let us now characterize solutions with such a $\beta^z$, they satisfy:
\[ \normA{\hat x-\hat z} + \normA{\hat y-\hat z} \leq 2 \beta^z - \beta^x - \beta^y = \normA{\hat x-\hat y}, \]
thus implying that $\normA{\hat x-\hat z} + \normA{\hat y-\hat z} = \normA{\hat x-\hat y}$, which is equivalent to (\ref{minmax}).
Also, these solutions are such that $\normA{\hat x-\hat z} = \beta^z - \beta^x$ and $\normA{\hat y-\hat z} = \beta^z - \beta^y$.
Thus, there exist solutions only is $\beta^z - \beta^x \geq 0$ and $\beta^z - \beta^y \geq 0$, and the 
combination of these two equalities, with $\beta^z$ defined by (\ref{beta_opt}), is 
equivalent to (\ref{exist_mub}).

Let us now check that there exist minimal upper bounds under this assumption : we must prove that if (\ref{exist_mub}) 
holds, there exists $\hat z$ satisfying (\ref{beta_opt}) and (\ref{minmax}) such that 
\[ \normA{\hat x-\hat z} - \frac12 \normA{\hat x-\hat y} = \frac12 \normA{\hat x-\hat y} - \normA{\hat y-\hat z} = \frac12 (\beta^y - \beta^x) .\]
First part of this equality is always satisfied when (\ref{minmax}) holds. Second part is about the existence of 
solutions to $f(\hat z) = 2\normA{\hat y-\hat z}- \normA{\hat x-\hat y} + \beta^y - \beta^x = 0$. Using (\ref{exist_mub}), we have 
$f(\hat y) \leq 0$ and $f(\hat x) \geq 0$, so there exists indeed such minimal upper bounds $\hat z$ when (\ref{exist_mub}) is satisfied.
$\Box$
\end{proof}

\begin{exem}
\label{ex_u_0}
Take $\hat x = \varepsilon_1$ and $\hat y= 2 \varepsilon_U$, condition (\ref{exist_mub}) is not satisfied, so there 
exists no minimal upper bounds. Indeed, minimal upper bounds would be $\hat z = a + b \varepsilon_1 + 1.5$, 
with $0 \leq a \leq 0$ and $0 \leq b \leq 1$. And expressing $\normA{x-z}=\beta^z = 1.5$ gives $b=-0.5$, which 
is not admissible (not in [0,1]).
\end{exem}
We note that when $\hat x$ and $\hat y$ do not have $\varepsilon_U$ symbols, there always 
exist minimal upper bounds. In the case when they do not exist, we will use a widening 
introduced in Definition \ref{widen_1}.

\begin{exem}
\label{ex_u_1}
Take $\hat x = 1 + \varepsilon_1$ and $\hat y = 2 \varepsilon_1$. We have $\gamma(\hat x) = [0,2]$ and 
$\gamma(\hat y) = [-2,2]$, so $\hat x$ and $\hat y$ are in generic positions. Minimal 
upper bounds $\hat z$ of $\hat x$ and $\hat z$ are
$$\hat z = a + b \varepsilon_1 + \varepsilon_U,$$
where $\normA{\hat{x}-\hat{z}} = \normA{\hat{y}-\hat{z}} = \beta^z$. 
This implies $a-b = -1$, with 
$0 \leq a \leq 1$, and $1 \leq b \leq 2$. Among these solutions, we find a unique one 
that minimizes the width of the concretization, by taking $b = 1$ and thus $a = 0$. 
This solution satisfies $\gamma(\hat z) = \gamma(\hat x) \cup \gamma(\hat y)$,
$\alpha_0^z (= a) = mid (\gamma(\hat x) \cup \gamma(\hat y))$ and
$\alpha_1^z (= b) =  \argmin{\alpha^x_1}{\alpha^y_1}$. 
In Proposition \ref{join2}, we show that this is 
a general result when $\hat x$ and $\hat y$ are in generic positions.
\end{exem}

\begin{exem}
\label{ex_u_2}
Now take $\hat x = 1 + \varepsilon_1$ and $\hat y = 4 \varepsilon_1$, this time $\gamma(\hat x)$ and 
$\gamma(\hat y)$ are not in generic positions. Minimal upper bounds $\hat z$ are: 
$$\hat z = a + b \varepsilon_1 + 2 \varepsilon_U,$$
where $a-b = -2$, $0 \leq a \leq 1$, and $1 \leq b \leq 4$. Now, let us
minimize the width 
of the concretization as in the previous example. The problem is that 
we cannot choose in this case
$b= \argmin{\alpha^x_1}{\alpha^y_1}=1$ because then the value of $a$ (-1) deduced from  $a-b = -2$ 
is not admissible (it is not beween 0 and 1).  
The solution minimizing  the width of the concretization is in fact
$\hat z = 2 \varepsilon_1 + 2 \varepsilon_U$, and it is such that
$\alpha_0^z (= a) = mid (\gamma(\hat x) \cup \gamma(\hat y))$ and 
$\gamma(\hat z) = \gamma(\hat x) \cup \gamma(\hat y)$. 
\end{exem}

We now give an intuition on the general case by taking examples with several noise symbols.
\begin{exem}
\label{ex_u_3}
Take $\hat x = 3 + \varepsilon_1 + 2 \varepsilon_2$ and $\hat y = 1 - 2 \varepsilon_1 + 
\varepsilon_2$. We have $\gamma(\hat x) = [0,6]$ and 
$\gamma(\hat y) = [-2,4]$, $\gamma(\hat x)$ and $\gamma(\hat y)$ are in generic positions.
 Minimal upper bounds are
$\hat z = a + b \varepsilon_1 + c \varepsilon_2 + 3 \varepsilon_U,$
where $a+b+c = 3$, $1 \leq a \leq 3$, $-2 \leq b \leq 1$, and $1 \leq c \leq 2$. 
Among these solutions, we can still find a unique one 
that minimizes the width of the concretization, taking $a = 2$, $b = 0$ and $c=1$: 
$\hat z = 2 + \varepsilon_2 + 3 \varepsilon_U$. 
\end{exem}

\begin{exem}
\label{ex_u_4}
Take $\hat x = 1 + \varepsilon_1 + 2 \varepsilon_2 + \varepsilon_3$ and 
$\hat y = -2 - 6 \varepsilon_1 + \varepsilon_2 + 2 \varepsilon_3$. 
We have $\gamma(\hat x) = [-3,5]$ and $\gamma(\hat y) = [-11,7]$, so here
$\gamma(\hat x)$ and $\gamma(\hat y)$ are not in generic positions.
 Minimal upper bounds are
$\hat z = a + b \varepsilon_1 + c \varepsilon_2 + d \varepsilon_3 +  6 \varepsilon_U,$
where $a+b+c-d=-3$, $-2 \leq a \leq 1$, $-6 \leq b \leq 1$, $1 \leq c \leq 2$ and $1 \leq d \leq d$. 
Again, as in Example \ref{ex_u_2}, minimizing the concretization of $\gamma_z$ by minimizing the 
absolute value of $b$, $c$ and $d$ does not give an admissible solution (when $b = 0$, $c=1$, $d=1$, 
relation $a+b+c-d=-3$ gives $a=-3$ which is not admissible). 
But, as the minimal concretization 
for $\hat z$ is in any case $\gamma(\hat x) \cup \gamma(\hat y)$, we can try to impose it. This 
gives $a = mid (\gamma(\hat x) \cup \gamma(\hat y)) = -2$, and an additional relation 
$-2+|b|+c+d+6=7$.\\
All solutions 
$\hat z = -2 + b \varepsilon_1 + c \varepsilon_2 + d \varepsilon_3 +  6 \varepsilon_U,$
with  $-6 \leq b \leq 1$, $1 \leq c \leq 2$, $1 \leq d \leq 2$, $|b|+c+d=3$ and $b+c-d=-1$ 
are minimal upper bounds with minimum concretization, and there are an infinite number of them, we
can choose for example $\hat z = -2 - \varepsilon_1 + \varepsilon_2 + \varepsilon_3 +  
6 \varepsilon_U,$ or $\hat z = -2 + \varepsilon_2 + 2 \varepsilon_3 +  
6 \varepsilon_U,$, etc.
\end{exem}

\begin{prop}
\label{join2}
Let $\hat x, \hat y \in \RA_1$, such that (\ref{exist_mub}) holds. 
If $\gamma(\hat x)$ and 
$\gamma(\hat y)$ are in generic positions, 
then $\hat z$ defined by (\ref{beta_opt}) and
\[ \left\{ 
\begin{array}{l} 
\alpha_0^z = mid (\gamma(\hat x) \cup \gamma(\hat y)) \\
\alpha_i^z = \argmin{\alpha^x_i}{\alpha^y_i}, \; \forall i \geq 1 
\end{array}
 \right.\] is the unique minimal upper bound of $\hat x$ and $\hat y$ whose concretization is 
the union of the concretization of $\hat x$ and $\hat y$.

If $\gamma(\hat x)$ and $\gamma(\hat y)$ are not in generic positions and 
$\gamma(\hat x) \subset \gamma(\hat y)$ (we get symmetric properties when 
$\gamma(\hat y) \subset \gamma(\hat x)$), then all $\hat z$ satisfying (\ref{beta_opt}), 
(\ref{minmax}), and
\[ \alpha_0^z = mid (\gamma(\hat x) \cup \gamma(\hat y)) = \alpha_0^y \]
\begin{equation}
\alpha_i^y \leq \alpha_i^z \leq 0 \mbox{ or } 0 \leq  \alpha_i^z  \leq \alpha_i^y , \; \forall i \geq 1
\label{min}
\end{equation}
are minimal upper bounds with concretization the union of the concretization 
of $\hat x$ and $\hat y$.
\end{prop}

Indeed, the solution with minimal concretization is particularly 
interesting when computing fixpoint in loops, by preserving the stability of the concretizations of 
variables values in iterates, as we will see in Theorem \ref{thm}.

\begin{sketch}
We want to find $\alpha_i^z$ such that the concretization is the smallest possible, with the above 
conditions still holding. For that, we have to minimize $|\alpha_i^z|$ with constraints (\ref{minmax}), 
we thus set \[ \alpha_i^z = \argmin{\alpha^x_i}{\alpha^y_i}, \; \forall i \geq 1 \]

For this choice of the $\alpha_i^z$, then for all $i \geq 1$, we can prove the two following properties:

\begin{equation} 
\label{prop1}
|\alpha_i^z - \alpha_i^x| - |\alpha_i^z - \alpha_i^y| = |\alpha_i^x| - |\alpha_i^y|,
\end{equation}
\begin{equation}
\label{prop2}
|\alpha_i^z| = \frac12 (|\alpha_i^x|+|\alpha_i^y|-|\alpha_i^y-\alpha_i^x|).
\end{equation}
We still have to define $\alpha_0^z$: let us now, using $\normA{\hat x-\hat z} = \beta^z - \beta^x$ and 
$\normA{\hat y-\hat z} = \beta^z - \beta^y$, write 
$\normA{\hat x-\hat z} - \normA{\hat y-\hat z} = \beta^y - \beta^x$ and express it using property (\ref{prop1}). \\
When $\alpha_0^x \leq \alpha_0^z \leq \alpha_0^y$, we can then show that it can be rewritten as
\begin{equation}
\label{alp0} 
\alpha_0^z = \frac12 (\alpha_0^x + \alpha_0^y + \sum_{i \geq 1} |\alpha_i^y| + \beta^y - 
\sum_{i \geq 1} |\alpha_i^x|  - \beta^x ),
\end{equation}
and, using (\ref{prop2}), that 
\[ \alpha_0^z + \sum_{i \geq 1} |\alpha_i^z| + \beta^z = \alpha_0^y + \sum_{i \geq 1} |\alpha_i^y| + 
\beta^y. \]
When $\alpha_0^x \leq \alpha_0^y$, and $\gamma(\hat x)$ and $\gamma(\hat y)$ are in generic positions,
then $\alpha_0^x - \sum_{i \leq 1} |\alpha_i^x| - \beta^x$ is the minimum of 
$\gamma(\hat x) \cup \gamma(\hat y)$, 
and $\alpha_0^y + \sum_{i \leq 1} |\alpha_i^y| + \beta^y$ its maximum. So $\alpha_0^z$ is indeed the center of 
$\gamma(\hat x) \cup \gamma(\hat y)$, and the concretization of $\hat z$ thus defined is the minimal 
possible, that is $\gamma(\hat x) \cup \gamma(\hat y)$. 
The proof is of course symmetric when $\alpha_0^y \leq \alpha_0^x$.

Now if $\gamma(\hat x)$ and $\gamma(\hat y)$ are not in generic positions, and for instance here 
$\gamma(\hat x) \subset \gamma(\hat y)$, then we can use $\alpha_0^y + \sum_{i \geq 1} |\alpha_i^y| + \beta^y > 
\alpha_0^x + \sum_{i \geq 1} |\alpha_i^x| + \beta^x $in (\ref{alp0}) to deduce $\alpha_0^z > \alpha_0^x$ and 
$\alpha_0^x - \sum_{i \geq 1} |\alpha_i^x| - \beta^x > \alpha_0^y - \sum_{i \geq 1} |\alpha_i^y| - \beta^y$ 
in (\ref{alp0}) to deduce $\alpha_0^z > \alpha_0^y$, which is not admissible.
So $\hat z$ given by $\alpha_i^z = \argmin{\alpha^x_i}{\alpha^y_i}$ is not a minimal upper bound in the non 
generic case. In order to have minimal concretization, we must have 
$\alpha_0^z = mid (\gamma(\hat x) \cup \gamma(\hat y)) = \alpha_0^y$, and 
\[ \alpha_0^z + \sum_{i \geq 1} |\alpha_i^z| + \beta^z = \alpha_0^y + \sum_{i \geq 1} |\alpha_i^y| + \beta^y, \]
which can be rewritten \[\normL{z} = \normL{y} + \beta^y - \beta^z = \normL{y} - \normL{y-z}, \]
equivalent to (\ref{min}). The proof and conditions are of course symmetric when 
$\gamma(\hat y) \subset \gamma(\hat x)$.
$\Box$
\end{sketch}

Note that, as we will show in Section \ref{riesz}, the join operator thus defined is not associative 
in the non generic case. What's more, as we will see, the affine form obtained by two successive 
join operations may not even be a minimal upper bound of the three joined affine forms. In Section, 
\ref{riesz}, we will thus introduce a first (associative) widening\footnote{This is a slight abuse of notation here: we do not have in general the
finite chain property, but a similar one, in our framework (convergence
in a finite time, in finite arithmetic).} of this join operation, which we 
will use to define a partial order.

\subsubsection{The meet operation}

\label{meetoperation}

If $(\RA_1,\leq)$ admitted binary least upper bounds, then we would have
a Riesz space, for which $x \cap y$ would be defined as $-\left((-x)\cup
(-y)\right)$. Here, we have a different formula, linking $\cap$ with $\cup$
in some interesting cases. Intersections will produce {\em negative} $\beta$
coefficients, where unions were producing {\em positive} $\beta$ coefficients.

\begin{lemm}
\label{meetbeta}
For all $\hat{x}$, $\hat{y}$ in $\RA_1$, there exist {\em maximal lower
bounds} (or mlb) $\hat{z}$ of $\hat{x}$ and $\hat{y}$
if and only if
\begin{equation}
\label{exist_mlb}
\normA{\hat{y}-\hat{x}} \geq | \beta^y - \beta^x| .
\end{equation}

They then all satisfy (for all $i\geq 1$):
$$\beta^z= \frac{1}{2}\left(\beta^x+\beta^y-\normA{y-x}\right)$$
$$\alpha^x_i \wedge \alpha^y_i \leq \alpha^z_i \leq \alpha_i^x \vee \alpha^y_i$$
\end{lemm}

\begin{sketch}
Being a lower bound of $x$ and $y$ means:
$$\begin{array}{rcl}
\normA{z-x} & \leq & \beta^x-\beta^z \\
\normA{z-y} & \leq & \beta^y-\beta^z \\
\end{array}$$
Summing the two inequalities, and using the triangular inequality:
$$\begin{array}{rcl}
\normA{y-x} & \leq & \normA{z-y}+\normA{z-x} \\
& \leq & \beta^x+\beta^y-2\beta^z \\
\end{array}$$
Hence 
$\beta^z \leq \frac{1}{2}\left(\beta^x+\beta^y-\normA{y-x}\right)$
giving an upper bound. As we want a maximal $z$, the natural question is
whether we can reach this bound. This is the case when the triangular
inequality for $\ell_1$ norm is an equality, which is the case when
$\alpha^x_i \wedge \alpha^y_i \leq \alpha^z_i \leq \alpha_i^x \vee \alpha^y_i$.
A solution exists to these constraints only if (\ref{exist_mlb})
is satisfied, as for the proof of Proposition \ref{join1}. $\Box$
\end{sketch}

Contrarily to the join operators, we cannot in general impose (even
in generic position) for a 
mlb $\hat{z}$ to have a given concretization, such as $\gamma(\hat{x})
\cap \gamma(\hat{y})$, or even a smaller value, such as the interval
that contains all values that $\hat{x}(t_1,\ldots,u_x)$ and 
$\hat{x}(t_1,\ldots,u_y)$ share for some $t_1,\ldots, u_x,u_y \in [-1,1]$. 

\begin{exem}
Consider $\hat{x}=1+\varepsilon_1-2\varepsilon_2 \subseteq [-2,4]$
and $\hat{y}=2+2\varepsilon_1+\varepsilon_2 \subseteq [-1,5]$. They are
in generic position. Then 
$\hat{z}$ is a mlb with concretization $\gamma(\hat{z})=\gamma(\hat{x})
\cap \gamma(\hat{y})=[-1,4]$ if and only if
$\hat{z} = 1.5+a \varepsilon_1+b \varepsilon_2-\frac{5}{2} \varepsilon_U$, 
$a+b = 1$, 
$a+|b| = 5$, 
$-2 \leq b \leq 1$ and
$1 \leq a \leq 2$.
Suppose $b$ is positive, then we want to have $a+b=5$ and $a+b=1$, which
is impossible. So $b$ is negative and we want to solve $a-b=5$ and $a+b=1$,
therefore, $a=3$ and $b=-2$, which is impossible because we precisely
asked $b$ for being positive (and $a$ to be less than 2).
\end{exem}

In some cases though, such mlb operators exist and we can give an explicit
formula:

\begin{lemm}
\label{argmax}
In case $\hat{x}$ and $\hat{y}$ are in generic positions, (\ref{exist_mlb}) is satisfied, and
$\alpha^x_i \alpha^y_i \geq 0$ for all $i \geq 1$, there exists
a maximal lower bound $\hat{z}$ with $\gamma(\hat{z})=\gamma(\hat{x})
\cap \gamma(\hat{y})$, given by the formulas:
\begin{itemize}
\item $\alpha^z_0=mid\left( \gamma(\hat{x})\cap \gamma(\hat{y}) \right)$
\item $\alpha^z_i=\argmax{\alpha^x_i}{\alpha^y_i}$ for all $i \geq 1$
\item $\beta^z= \frac{1}{2}\left(\beta^x+\beta^y-\normA{y-x}\right)$
\end{itemize}
In this case, we have:
\begin{equation}
\label{spectral}
x\cap y+x\cup y=x+y
\end{equation}
\end{lemm}

\begin{proof}
The formula for $\beta^z$ is given by Lemma \ref{meetbeta}. The fact
that the concretization of $\hat{z}$ is $\gamma(\hat{x})\cap \gamma(\hat{y})$
implies the formula for $\alpha^z_0$. 

The formulas for $\alpha^z_i$, $i \geq 1$ can be checked easily as follows:
as $\gamma(\hat{x})$ and $\gamma(\hat{y})$ are in generic positions, 
$\gamma(\hat{x}) \cap \gamma(\hat{y})$ and $\gamma(\hat{x})$ (similarly
with $\gamma(\hat{y})$) are in generic positions. Thus we can use the
formula of Proposition \ref{join2} for the join operator, to compute
$z \cup x$ and $z\cup y$. 
It is easily seen now that for $uv \geq 0$, and $w$,
$$\argmin{\argmax{u}{v}}{u}=\argmin{u}{v}$$  $$\argmin{\argmax{u}{v}}{v}=
\argmin{u}{v}$$
Hence $\alpha^{z\cup y}_i=\alpha^y_i$ and $\alpha^{z\cup x}_i=\alpha^x_i$
for all $i\geq 1$. 

Furthermore, $\gamma(\hat{z}\cup \hat{y})=\gamma(\hat{y})$ and
$\gamma(\hat{z}\cup \hat{x})=\gamma(\hat{x})$ hence $\alpha^{z \cup y}=
\alpha^y$ and $\alpha^{z \cup x}=\alpha^x$. Finally, again because of this
equality and concretizations, and that all coefficients but $\beta^{z \cup x}$
(respectively $\beta^{z \cup y}$) have been shown equal to the ones
of $\hat{x}$ (respectively $\hat{y}$), we have necessarily that
$\beta^{z \cup x}=\beta^x$ (respectively $\beta^{z \cup y}=\beta^y$). 

Therefore $x=z\cup x \geq z$ and $y=z\cup y \geq z$ so $z$ is a lower
bound of $x$ and $y$. Because of the value of $\beta^z$, by Lemma
\ref{meetbeta} and \ref{necmin} (adapted to mlbs), $z$ is an mlb (with
the right concretization). $\Box$
\end{proof}

\subsection{Quasi bounded completeness}
\label{riesz}

We prove that we have almost bounded completeness of $\RA_1$. Unfortunately,
as shown in Example \ref{nonassoc}, this is barely usable in practice,
and we resort to a useful sub-structure of $\RA_1$ in Section
\ref{widen} (in particular, with a view to Section \ref{convsection}).

\begin{prop}
$(\RA_1,\leq)$ is a quasi bounded-complete partial order (or is quasi-``Dedekind-complete''), meaning that any bounded subset $A$ of $\RA_1$ such that
for all $\hat{x}$, $\hat{y}$ in $A$
\begin{equation}
\label{quasibounded}
\normA{\hat{x}-\hat{y}} \geq |\beta^x-\beta^y|
\end{equation}
has a
minimal upper bound in $\RA_1$.
\end{prop}

\begin{sketch}

Let $z$ be an upper bound of all $a_i \in A$. This means again, for all
pairs $i,j$ that
$$\beta^z \geq \frac{1}{2}\left(\normA{a_i-a_j} + \beta^{a_i}+\beta^{a_j}\right)$$
We can always suppose that $b=\alpha^b_0+\beta^b \varepsilon_U$. As $b$
dominates all $a_i \in A$, we have, using the triangular inequality:
$$\begin{array}{rcl}
\normA{a_i-a_j} & \leq & \normA{a_i-b}+\normA{b-a_j} \\
& \leq & 2\beta^b-\beta^{a_i}-\beta^{a_j} \\
\end{array}$$
so $\beta^z \geq \beta^b$. This means that we can write:
$$\beta^z = \inf_{i,j} \left\{\frac{1}{2}\left(\normA{a_i-a_j} + \beta^{a_i}+\beta^{a_j}\right)\right\}$$
which exists in $\R$. 
Similarly to the proof of Proposition \ref{join1}, condition \ref{quasibounded}
allows to prove existence of a solution to the mub equations.

$\Box$
\end{sketch}

Unfortunately, even if we pick one of the possible join operators,
they are not in general associative operators, which means that even
for countable subsets $A$ of $\RA_1$, according to the iteration strategy 
we choose, we might end up with a non-minimal upper bound.

\begin{exem}
\label{nonassoc}
Take 
\[ \left\{
\begin{array}{ccl} 
\hat x &=& 1 + 2 \varepsilon_1 - \varepsilon_2 + 2 \varepsilon_3 \\
\hat y & = & \varepsilon_1 + \varepsilon_2 + \varepsilon_3 \\
\hat z & = & 5 + \varepsilon_1 - 2 \varepsilon_2
\end{array}
\right.
\]
When computing with one of the possible join operators previously defines, we obtain
\[ \left\{
\begin{array}{ccl} 
(\hat x \cup \hat y) \cup \hat z & = & 2 + \varepsilon_1 + b_1 \varepsilon_2 + 
0.5 (7.5 + a_1 + b_1 + c_1) \varepsilon_U \\
(\hat y \cup \hat z) \cup \hat x & = & 2 + \varepsilon_1 + 5 \varepsilon_U \\
(\hat x \cup \hat z) \cup \hat y & = & 2 + \varepsilon_1 + b_2 \varepsilon_2 + 
(0.5+\beta_2) \varepsilon_3 + 4.5 \varepsilon_U,
\end{array}
\right.
\]
with $a_1-b_1+c_1 = 2.5$, $1 \leq a_1 \leq 2$, $-1 \leq b81 \leq 0$, $1 \leq c_1 \leq 2$, and 
$-1 \leq b_2 \leq 0$.

We see on this example that, in the case where all concretizations are not in generic positions,
 the join operator is not associative. Moreover, the result of two successive 
join operations may not be a minimal upper bound of three affine forms because we do not always 
get the same $\beta$ coefficient (5 when computing $(\hat y \cup \hat z) \cup \hat x$, and 
4.5 when computing $(\hat x \cup \hat z) \cup \hat y)$. Indeed, we have here
$(\hat x \cup \hat z) \cup \hat y \leq (\hat y \cup \hat z) \cup \hat x$ when $-0.5 \leq b_2 \leq 0$.
\end{exem}

We fix this difficulty in next section.

\subsection{A bounded lattice sub-structure}
\label{widen}
In practice, we obtain a stronger sub-structure by using a widening instead 
of the minimal upper bound: 
\begin{deflem}
\label{widen_1}
We define the widening operation $\hat z = \hat x \nabla \hat y$ by
\begin{itemize}
\item $\alpha^z_0=mid\left( \gamma(\hat{x})\cup \gamma(\hat{y}) \right)$
\item $\alpha^z_i=\argmin{\alpha^x_i}{\alpha^y_i}\; $ for all $i \geq 1$,
\item $\beta^z= \sup \gamma(\hat{x}\cup \hat{y}) - \alpha^z_0-\normL{z}$
\end{itemize}
In the case when $\hat x$ and $\hat y$ are in generic positions,
$\nabla$ is the union defined in Section \ref{sec_join}. Otherwise, $\hat x \nabla \hat y$
 has as concretization the union of the concretizations of $\hat x$ and $\hat y$, it 
is an upper bound of $\hat x$ and $\hat z$ (but it is not a minimal upper bound with 
respect to $\leq$, because $\beta^z$ is not minimal).
\end{deflem}

This operator has the advantages of presenting a simple an explicit formulation, a stable 
concretization with respect to the operands, and of being associative.

\begin{deflem} [computational order] 
Let $\ll$ be the binary relation defined by:
$$\hat{x} \ll \hat{y} \Leftrightarrow \hat{x} \nabla \hat{y} =\hat{y}$$
Then, $\ll$ is a partial order.
\end{deflem}

\begin{sketch}
Reflexivity comes from $\hat{x} \nabla \hat{x}=\hat{x}$. Antisymmetry is trivial. Transitivity
comes from the associativity of $\nabla$. $\Box$
\end{sketch}

\begin{lemm}
$\hat{x}\ll\hat{y}$ if and only if:
\begin{itemize}
\item $\gamma(\hat{x}) \subseteq \gamma(\hat{y})$ and
\item for all $i\geq 1$, $0 \leq \alpha^y_i \leq \alpha^x_i$ or
$\alpha^x_i \leq \alpha^y_i \leq 0$
\end{itemize}
\end{lemm}

\begin{sketch}
The first condition ensures that $\alpha^{x \cup y}_0=\alpha^y_0$ and
$\beta^{x \cup y}=\beta^y$. Take $i \geq 1$. If $\alpha^x_i \alpha^y_i 
\leq 0$ then because $\argmin{\alpha^x_i}{\alpha^y_i}=\alpha^y_i$,
$\alpha^y_1$ has to be zero. Otherwise, this translates precisely to
the second condition.
\end{sketch}

\begin{defi}
We define operation $\hat z = \hat x \Delta \hat y$ by 
\begin{itemize}
\item $\alpha^z_0=mid\left( \gamma(\hat{x})\cap \gamma(\hat{y}) \right)$
\item $\alpha^z_i=\argmax{\alpha^x_i}{\alpha^y_i}$ for all $i \geq 1$, 
if $\alpha^x_i \alpha^y_i \geq 0$, otherwise $\alpha^z_i=0$
\item $\beta^z= \sup \gamma(\hat{x}\cap \hat{y}) - \alpha^z_0-\normL{z}$
\end{itemize}
\end{defi}

\begin{prop}
\label{completelattice}
$(\RA_1,\ll)$ is a bounded complete lattice, with:
\begin{itemize}
\item $\nabla$ being the union,
\item $\Delta$ being the intersection.
\end{itemize}
\end{prop}

\begin{sketch}
Easy verification for the binary unions and intersections. 
Take now $A \subseteq \RA_1$ and $b$ such that $b \geq A$. Then
$\gamma(a) \subseteq \gamma(b)$ for all $a \in A$, so $I=\cup_{a \in A}
\gamma(a)$ has finite bounds. This gives us $\alpha_0=mid(I)$. 

Consider now any countable 
filtration of $A$ by an increasing sequence of finite subsets
$A_k$, $k \in K$, and consider $a_k=\cup A_k$ any minimal upper bound of 
$A_k$. We know that $\mid \alpha^{a_k}_j \mid$ is a decreasing sequence
of positive real numbers when $k$ increases, so it converges. 
We also know that the sign of $\alpha^{a_k}_j$ remains constant, so 
$\alpha^{a_k}_j$ converges, say to $\alpha_j$. Last but not least, let
$\beta=\frac{\sup I - \inf I}{2}- \alpha_0 - \sum_{i=1}^{\infty} \mid \alpha_i \mid$,
we argue that $\alpha_0+\sum_{i=1}^{\infty} \alpha_i \varepsilon_i +\beta
\varepsilon_U$ is a minimal upper bound of $A$ in $\RA_1$.
$\Box$
\end{sketch}

This allows to use $\nabla$ (and $\Delta$) as effective widenings during
the iteration sequence for solving the least fixed point problem.

\section{Iteration schemes and convergence properties}

\label{convsection}

This is where all properties we studied fit together, to reach the
important Theorem \ref{thm}, stating good behavior of the Kleene-like
iteration schemes defined in Section \ref{iteration}. First, we show that
we must improve the computation of the abstract semantic functional, between
two union points, this is explained in Section \ref{shiftoperator}. We also
improve 
things a little bit, on the practical side, by defining new widening
operators, in Section \ref{realwiden}.

\subsection{The shift operator and the iteration scheme}

\label{shiftoperator}

One problem we encounter if we are doing the blind Kleene iteration in the lattice
of Proposition \ref{completelattice}, is
that we introduce $\varepsilon_U$
coefficients, for which the semantics of arithmetic expressions is far less
well behaved than for ``ordinary'' noise symbols $\varepsilon_i$. 

\begin{exem}
\label{exxmoinsax}
Let us give a first simple example of what can go wrong. Consider the
following program:
\begin{verbatim}
F(real a) {
  real x;
  x = input(-1,1); [1]
  while (true)
    x = x-a*x; [2] }
\end{verbatim}
Suppose that \texttt{a} can only be given values between (strictly) $0$ and
$1$, then it is easy to see that this scheme will converge towards zero,
no matter what the initial value of $x$ is. 
As the scheme is essentially equivalent, in
real numbers to $x_{n+1}=(1-a)x_n$, with $|1-a|<1$, a simple Kleene iteration
scheme should converge. Let us look at the successive iterates 
$\hat{x}_i$ at control point [2], of this scheme. First, note that
$\hat{x}_{i+1}=\varepsilon_1 \nabla (\hat{x}_i\hat{-}a\hat{x}_i)$
(or equivalently $\hat{x}_{i+1}=\hat{x}_i\nabla (\hat{x}_i\hat{-}a\hat{x}_i)$
starting with $\hat{x}_0=\varepsilon_1$), 
where
$\varepsilon_1$ stands for the noise symbol introduced by assignment at
control point [1]).
$$\begin{array}{rcl}
\hat{x}_{1} & = & \varepsilon_1 \\
\hat{x}_2 & = & (1-a)\varepsilon_1 + a \varepsilon_U \\
\hat{x}_3 & = & \varepsilon_1 \nabla \left((1-a)^2 \varepsilon_1+
a\varepsilon_U\hat{-}a^2\varepsilon_U\right) \\
& = & \varepsilon_1 \nabla \left((1-a)^2 \varepsilon_1+a(1+a)\varepsilon_U\right)
\end{array}$$
because the semantics of $\hat{-}$ on $\varepsilon_U$ symbols cannot cancel
out its coefficients {\em a priori}. We will see a bit later that under
some conditions, we can improve the semantics {\em locally}.

To carry on with this example, let us particularize the above scheme to
the case $a=\frac{3}{4}$:
$$\begin{array}{rcl}
\hat{x}_3 & = & \varepsilon_1 \nabla \left(\frac{1}{16}\varepsilon_1
+\frac{21}{16}\varepsilon_U \right) \\
& = & \frac{1}{16}\varepsilon_1+\frac{21}{16}\varepsilon_U 
\subseteq \left[-\frac{11}{8},\frac{11}{8}\right]
\end{array}$$
We already see that the concretization of $\hat{x}_3$ is bigger than $[-1,1]$
showing loss of precision, even to simple interval computations. The next iterations
make this interval grow to infinity. If we could have written 
\begin{equation}
\label{goodeq}
\hat{x}_{i+1}=\varepsilon_1 \nabla (\hat{x}_i-a\hat{x}_i)
\end{equation}
instead of \begin{equation}
\label{badeq}
\hat{x}_{i+1}=\varepsilon_1 \nabla (\hat{x}_i\hat{-}a\hat{x}_i)
\end{equation} the
iteration sequence would have been convergent. We are going to explain that we can make
sense of this.
\end{exem}

We first introduce a {\em shift} operator, that decreases the current 
abstract value.

\begin{deflem}
\label{shift}
Let 
$\hat{x}=\alpha_0^x + \sum_{i=1}^{\infty} \alpha_i^x \varepsilon_i
+\beta^x \varepsilon_U$. 
Define, for $\hat{x}$ with finitely many non-zero coefficients:
$$
\begin{array}{rcl}
!\hat{x} & = &\alpha_0^x+\sum_{i=1}^{\infty} \alpha_i^x \varepsilon_i
+\beta^x \varepsilon_f \\ 
\end{array}
$$
where $\varepsilon_f$ is a ``fresh'' symbol. 
Then for all $\hat{x} \in \RA_+$, $! \hat{x}\leq \hat{x}$.
\end{deflem}


The idea is that, after some unions, during a Kleene iteration sequence,
such as right after the union in semantic equation \ref{badeq}, of
Example \ref{exxmoinsax}, we would like to apply the $!$ operator, allowing
us to get an equation equivalent to (\ref{goodeq}). 

The full formalization of this refined iteration scheme is outside the
scope of this paper. It basically relies on the following observation of
our abstract semantics:

{\em All concrete executions of a program correspond to a unique choice
of values between -1 and 1, of $\varepsilon_1,\ldots,\varepsilon_n,\ldots,
\varepsilon_{U}$ (for all join control points). 
}

Hence, locally, between two join control points, $\varepsilon_U$ coefficients
act as normal $\varepsilon_i$ coefficients. Hence 
one can use $!$ to carry on the evaluation of the abstract functionals
after each control point where a union was computed, corresponding to
a branching between several concrete executions. 

\begin{exem}
We carry on with the example \ref{exxmoinsax}. 
We use now the new semantic equation:
$\hat{x}_{i+1}=\varepsilon_1 \nabla (!\hat{x}_i\hat{-}a!\hat{x}_i)$.
Therefore, the successive iterates, for $a=\frac{3}{4}$ are:
$$\begin{array}{rcl}
\hat{x}_{1} & = & \varepsilon_1 \\
\hat{x}_2 & = & \frac{1}{4}\varepsilon_1 + \frac{3}{4}\varepsilon_U \\
\hat{x}_3 & = & \varepsilon_1 \nabla \left(\frac{1}{16} \varepsilon_1+
\frac{3}{4}\varepsilon_2\hat{-}\frac{9}{16}\varepsilon_2\right) \\
& = & \varepsilon_1 \nabla \left(\frac{1}{16}\varepsilon_1+\frac{3}{16}\varepsilon_2\right)
\end{array}$$
where $\varepsilon_2=!\varepsilon_U$ in the last iterate. Therefore,
$\hat{x}_3 = \frac{1}{16}\varepsilon_1+\frac{15}{16}\varepsilon_U$.
The successive iterates will converge very quickly to $\hat{x}_\infty=
\varepsilon_U$ with concretization being $[-1,1]$ and no surviving relation.
\end{exem}


\subsection{Iteration schemes}

\label{imp}
\label{iteration}

As we have almost bounded completeness, and not unconditional completeness,
our iteration schemes will be parametrized by a large interval $I$: as soon
as the current iterate leaves $I$, we end iteration by $\top$ (that
we can choose to represent by $\infty \epsilon_U$ by an abuse of notation).
The starting abstract value of the Kleene like iteration of Definition
\ref{iter} is as usual $\bot$, that, in theory, we should formally introduce
in the {\em lifted} domain of affine forms (but that, by an abuse of
notation, we can represent by $-\infty \epsilon_U$).

In order to get good results, we need in particular {\em cyclic unfolds}. 
They are defined below:

\begin{defi}
\label{iter}
Let $i$ and $c$ be any positive integers, $\sqcup$ be any of the
operators $\cup$ (for some choice of a mub), or $\nabla$. 

The $(i,c,\sqcup)$-iteration scheme of some functional $F: 
\hat{\var} \rightarrow \hat{\var}$ is as follows:
\begin{itemize}
\item First unroll $i$ times the Kleene iteration sequence, starting
from $\bot$, i.e. compute $x_1 = F^i(\bot)$.
\item Then iterate: $x_{n+1}=x_n \cup F^c(! x_n)$ starting with $n=1$.
\item End when a fixpoint is reached or with $\top$ if $\gamma(x_{n+1})
\not \subseteq I$.
\end{itemize}
\end{defi}

Note that initial unfolding are important for better precision but will not be used 
in the sequel.

\subsection{Convergence for linear recursive filters}

\label{linear}

We prove that our approach allow us to find good
estimates for the real bounds of general affine recurrences (i.e. 
{\em linear recursive filters} of any order), see Section \ref{affine}. 
The only abstract domains
known to be able to give accurate results are the one of \cite{Feret}, 
which only deals with filters of order 2, and
the one of \cite{compositional}, which is specialized for digital filters 
(which is not the case of our abstraction).

We consider again the class of programs of Section \ref{affine}.

\begin{theo}
\label{thm}
Suppose scheme (\ref{LRF}) has bounded outputs, i.e. the (complex) roots
of $x^n-\sum_{i=0}^{n-1} a_{i+1} x^i$ have module strictly less than 1.
Then there 
exists $q$ such that the $(0,q,\nabla)$-iteration scheme 
(see Section \ref{iteration})
converges towards a finite over-approximation of the output.

In other words, the perturbed numerical scheme solving the fixpoint
problem is also bounded.
\end{theo}

\begin{sketch}
Being a fixed point of abstract functional $F$ (giving the abstract
semantics of the one iteration of the loop) means
$$\begin{array}{rcl}
x_{k+n+1} & = & x_{n+1} \nabla \left(\sum_{i=1}^n a_i x_{k+i}+
\sum_{j=1}^{n+1} b_j e_{k+j}\right) \\
x_{n+1} & = & x_1 \nabla x_2 \nabla \ldots \nabla x_n \nabla x_{k+n+1} \\
x_{n+2} & = & x_2 \nabla \ldots \nabla x_n \nabla x_{k+n+1} \\
\ldots \\
x_{k+n} & = & x_n \nabla x_{k+n+1} 
\end{array}$$
Define
$$\begin{array}{lcl}
y_1 & = & x_{n} \\
y_2 & = & x_{n-1} \nabla x_n \\
\ldots \\
y_n & = & x_1 \nabla x_2 \nabla \ldots x_n \\
y_{n+1} & = & x_{n+1} \\
\end{array}$$
Then the fixpoints of $F$ are determined by the fixed points $z$:
\begin{eqnarray}
\label{fixpointeq}
z & = & x_{n+1} \nabla \left(\sum_{i=1}^n a_i (y_{n+1-i} \nabla z) +
\sum_{j=1}^{n+1} b_j e_{k+j} \right).
\end{eqnarray}

Suppose first that $\sum_{i=1}^{n} | a_i | < 1$. Consider now the
interval fixpoint equation resulting from (\ref{fixpointeq}). As
$\gamma$ commutes with $\nabla$, by definition, and because
of Lemma \ref{base}, it transforms into
$$\begin{array}{rcl}
\gamma(z) & \subseteq & \gamma(x_{n+1}) \cup \left(\sum_{i=1}^n \left(a_i \gamma(y_{n+1-i}) \cup \gamma(z) \right) \right.\\
& & \left.+ \sum_{j=1}^{n+1} b_j \gamma(e_{k+j})\right).
\end{array}$$
This equation shows that $\gamma(z)$ 
is a pre-fixpoint of the interval abstraction
of our linear scheme. It is well known that in the case 
$\sum_{i=1}^{n} | a_i | < 1$, this interval abstraction admits a bounded
least fixpoint $z^I$. Hence, $z$ in this case is bounded by $z^I$
(for order $\leq$, when $z^I$ is 
written as $mid(z^I)+dev(z^I) \varepsilon_U$,
with $dev([a,b])=\frac{b-a}{2}$), 
hence has finite concretization. In fact, not only $z$ but
all the ascending sequence of the $(0,1,\nabla)$-iteration scheme from 
$\bot$ is bounded by $z^I$. Note that any ascending sequence for
any $(p,q,\nabla)$-iteration scheme is also ascending for the partial
order $\ll$. By Proposition \ref{completelattice}, it has a least upper bound,
which is the least fixed point of $F$ for partial order $\ll$, because
of the obvious continuity (in the $\ell_1$ sense) of $F$. Hence again, 
this fixed point is bounded by $z^I$ so has finite concretization.

Secondly, if the roots of $x^n-\sum_{i=0}^{n-1} a_{i+1} x^i$ have module strictly less than 1, 
then there exists $q$ such that $F^q$ is a filter of order
$nq$ in the inputs $e$, and $n$ in the outputs with coefficients $c_j$, 
$j=1,\ldots,n$ such that $\sum_{i=1}^{n}
|c_j|$ is strictly less than 1. One can check that the semantics on
affine forms is exact on affine computations (because of the use of
the shift operator). We can then
apply the result above to reach
the conclusion. $\Box$
\end{sketch}

More generally, and this is beyond the scope of this article, we can show
that there exist $(i,q,\nabla)$-iteration schemes that will come 
{\em as close
as we want} to the exact range of values that $x$ can take. 

\begin{exem}
We carry on with Example \ref{transf2}. 
We see that matrix $A$ in our case is $$A=\left(\begin{array}{cc}
0 & 1 \\
-0.7 & 1.4 
\end{array}\right)$$
and of course, the $\ell_1$ norm of the rows of $A$ is bigger than 1. 
Iterating $A$, we see that:
$$
%
A^5 = \left( \begin{array}{cc}
-0.5488 & 0.2156 \\
-0.15092 & -0.24696
\end{array} \right)
$$
is the first iterate of $A$ with $\ell_1$ norm of rows less than 1. 
We know by Theorem \ref{thm} that a $(0,5,\nabla)$-iteration scheme
will eventually converge to an upper approximation of the invariant
(which we can estimate, see Example \ref{transf22}, to 
[-1.12124069...,2.82431841...]). 
Here is what $(0,i,\nabla)$-iteration schemes give as last
invariant and concretization, when $i$ is greater than 5
(rounded, for purposes of readability):

\begin{center}
\begin{tabular}{|c|c|c|}
\hline
i & invariant & concretization \\
\hline
5 & 0.8333+2.4661$\varepsilon_U$ & [-1.6328,3.2995] \\
\hline
\hline
16 & 0.7621+2.06212$\varepsilon_U$ & [-1.3,2.8244] \\
\hline
\end{tabular}
\end{center}

Note that although the convergence to this invariant is {\em asymptotic}
(meaning that we would need in theory an infinite Kleene iteration 
to reach the invariant), in {\em finite precision}, the limit is
reached in a finite number of steps. In the case of the
$(0,16,\nabla)$-iteration scheme, the fixpoint is reached after 18 iterations.
In some ways, we have replaced the numerical scheme (a filter of order
2 here), by an {\em abstract numerical scheme} which has similar convergence
properties, and can be simulated in a finite time and in a {\em guaranteed
manner}, accurately. We can also
use {\em extrapolation} or {\em widening} techniques, for which we will
show some results in Example \ref{usewidening}.

Note also that none of the noise symbols survived in the final invariant: 
there is no dependency left with the successive inputs, when looking at
the overall invariant. This is very easily shown on the first few Kleene
iterates already. We denote by $\hat{x}_i$ the affine form at control
point [2], $\hat{y}_i$ the affine form at control point [3], at iteration $i$, for the $(0,16,\nabla)$-iteration scheme. 
We have (as produced by our prototype implementation):
$$\begin{array}{rcl}
\hat{x}_1 & = & 0.8808+1.8593\varepsilon_U \\
\hat{y}_1 & = & 0.8808+0.01038
\varepsilon_1+0.0429\varepsilon_2+0.0369 \varepsilon_4 +\ldots\\
& & +
0.1052\varepsilon_{21}+0.2046\varepsilon_{23}+0.2589\varepsilon_{25}+
0.2254\varepsilon_{27}\\
& & +0.081\varepsilon_{29}-0.16\varepsilon_{31}+
0.35\varepsilon_{33} \\
\hat{x}_2 & = & 
0.8422+1.9407 \varepsilon_U 
\\
\hat{x}_3 & = & 
0.8323+1.9688 \varepsilon_U \\
\ldots
\end{array}$$

Finally, you should note that Theorem \ref{thm} is not limited by
any means to finding invariants of such filter programs with independent
inputs, or independent initial conditions. For instance, if all the inputs
over time are equal, but unknown numbers between 0 and 1, the final invariant
has concretization [-0.1008,2.3298].

\end{exem}

\subsection{Simple widening operators}
\label{realwiden}

We can define numerous widening operators, among which the following:
\begin{deflem}
The operator $W$ defined by $\hat z = \hat x W \hat y$ such that
\begin{itemize}
\item $\alpha^z_0=mid\left( \gamma(\hat{x})\cup \gamma(\hat{y}) \right)$,
\item $\alpha^z_i=\alpha^x_i=\alpha^y_i$ for all $i \geq 1$ such that $\alpha^x_i=\alpha^y_i$,
\item $\alpha^z_i= 0$ for all $i \geq 1$ such that $\alpha^x_i \neq \alpha^y_i$,
\item $\beta^z= \sup \gamma(\hat{x}\cup \hat{y}) - \alpha^z_0-\normL{z}$
\end{itemize}
gives an upper bound of $\hat x$ and $\hat y$ that can be used 
as an efficient widening.
\end{deflem}

\begin{exem}
\label{usewidening}
Now we are carrying on with Example \ref{transf2}, but this time
we apply the widening defined above after 1 normal
iteration step. For $i$  equal to 5, fixpoint is reached at iteration 9, and for 
$i$ equal to 16, it is reached at iteration 4, with precision equivalent to the case 
without widening.
This time, convergence is reached in finite time, by construction (and
not because of ``topological'' convergence).
\end{exem}
\section{Directions currently investigated}
\label{improv}
We discuss in this section very promising improvements of the above schemes, which we feel are 
important to mention here. But, as they are not fully formalized yet, we mostly demonstrate them 
on examples.
\label{betaneg}

\subsection{Iteration strategies for a refined join operation}

As we introduced (Definition \ref{shift}) the possibility to shift the union symbols 
to ``classical'' noise symbols in the iteration scheme, it becomes important to create as
few union symbols as possible, in order not to lose relations. 
This can be partly solved by an adapted refined iteration strategy : when there is 
a cycle of explicit dependency between variables, make the union only on one variable and
apply immediately the shift operator of Definition \ref{shift}, before this union is propagated 
to the other dependent variables.
\begin{exem}
Consider the following program, implementing a second-order filter (where xnp1 stands for $x_{n+1}$, 
xn stands for $x_{n}$ and xnm1 stands for $x_{n-1}$):
\begin{verbatim}
  real xnp1,xn,xnm1;
  xn = [0,1];
  while (true) {
    xnp1 = 1.2*xn - 0.8*xnm1;
    xnm1 = xn; xn = xnp1;    }
\end{verbatim}
In this program, we have $x_n^k = x_{n-1}^{k+1}$, where $k$ is the current iteration of the loop, 
so it is clearly a bad idea to make unions independently on $x_n$ and $x_{n-1}$.

Before the loop, $\hat{x}_n^0 = 0.5 + 0.5 \varepsilon_1$, and after first iteration of the loop,
$\hat{x}_{n-1}^1 = 0.5 + 0.5 \varepsilon_1$, $\hat{x}_n^1 = \hat{x}_{n+1}^1 = 0.6 + 0.6 \varepsilon_1$.

Applying the classical join operations to $\hat{x}_n$ and $\hat{x}_{n-1}$ at the beginning of the loop after first 
iteration gives
$\bar x_n^1 = \hat x_n^0 \cup \hat x_n^1 = 0.6 + 0.5 \varepsilon_1 + 0.1 \varepsilon_U$,
$\bar x_{n-1}^1 = \hat x_{n-1}^0 \cup 0 = 0.5 + 0.5 \varepsilon_U$.\\
Then, after applying the shift operator (with a new symbol $\varepsilon_2$ for $\bar x_n^1$, and 
a new symbol $\varepsilon_3$ for $\bar x_{n-1}^1$), we get $\hat x_n^2 = \hat x_{n+1}^2 = 0.32 + 0.6 \varepsilon_1 + 
0.12 \varepsilon_2 - 0.4 \varepsilon_3$, and $\gamma(\hat x_n^2) = [-0.8,1.44]$.

Then $\bar x_n^2 = \hat x_n^2 \cup !(\bar x_n^1) = 0.32 + 0.5 \varepsilon_1 + 0.1 \varepsilon_2 + 0.52 \varepsilon_U$,
 $\bar x_{n-1}^2 = \hat x_{n-1}^2 \cup !(\bar x_{n-1}^1) = !(\bar x_n^1) \cup !(\bar x_{n-1}^1) = 0.6 + 0.6 \varepsilon_U$,
and after applying a shift that creates new symbols $\varepsilon_4$ and $\varepsilon_5$, we get 
$x_n^3 = x_{n+1}^3 = -0.096 + 0.6 \varepsilon_1 + 0.12 \varepsilon_2 + 0.624 \varepsilon_4 - 0.48 \varepsilon_5$ 
and $\gamma(\hat x_n^3) = [-1.92,1.728]$.

Of course, in practice, cyclically unrolling the loop allows to 
care with the bad behavior of the scheme, but it is better to refine as well the iteration as follows.
 
Applying the join and shift operations on $x_n$ only, we write 
$\bar x_n^1 = !(\hat x_n^0 \cup \hat x_n^1) = 0.6 + 0.5 \varepsilon_1 + 0.1 \varepsilon_2$, and 
$\hat x_n^2 = 0.32 + 0.2 \varepsilon_1 + 0.12 \varepsilon_2$ , and $\gamma(\hat x_n^2) = [0,0.64]$.

Then $\bar x_n^2 = !(\hat x_n^1 \cup \hat x_n^2) = 0.6 + 0.2 \varepsilon_1 + 0.1 \varepsilon_2 + 0.3 \varepsilon_3$ and
$\hat x_n^3 = 0.24 - 0.16 \varepsilon_1 + 0.04 \varepsilon_2 + 0.36 \varepsilon_3$, $\gamma(\hat x_n^3) = [-0.32,0.8]$.
\end{exem}

We dealt here with an example where the dependencies between variables were explicit, but we can also generalize 
this and introduce symbols to explicit and preserve implicit dependency, as in Example \ref{ex_improve_2}:
\begin{exem}
\label{ex_improve_2}
Take the following program:
\begin{verbatim}
real x,y;
x = 1; y = 0;
for (x=1; x<=10000 ; x++)
  y = y + 2;
\end{verbatim}
If we apply standard union, we get for example after first iteration, union and shift :
$\bar x^1 = 1.5 + 0.5 \varepsilon_1$ and $\bar x^2 = 1 + \varepsilon_2$. 

We can not use the above strategy for using union on one variable only, because 
variables $x$ and $y$ are no explicitly linked. 

Still, we can observe that, before or inside the loop, $x$ and $y$ are always such that $y = 2x-2$. 
So, in order to construct  a joint union on $\hat x$ and $\hat y$ that preserves this relation, we can just use union 
and shift on $\hat x$ ($\bar x^1 = 1.5 + 0.5 \varepsilon_1$), and then simply deduce $y$ by 
$\hat y = 2\hat x-2= 1 + \varepsilon_1$, thus expressing relations that will be usable in further computation. 
And of course the same operation can be repeated again for the following iterations of the loop, as affine
relations are preserved by affine arithmetic :
here $\hat x^2 = \bar x^1 + 1 = 2.5 + 0.5 \varepsilon_1$ and $\hat y^2 = \hat y^1 + 2 = 3 + \varepsilon_1$, we still have 
$\hat y^2 = 2\hat x^2-2$.

And indeed, for practical realization, when making unions over a set of variables, it is easy for example 
to consider couple of variables, and first investigate whether or not they satisfy the same affine relation 
on the two states, and if it is the case propagate the union over one variable on the other.
\end{exem}

\subsection{Refining again the join operator for disjunctive analysis}

\begin{exem}
Take the following program:
\begin{verbatim}
real x,y,z;
z = [0,1];
if (z < 0.5) {
  x = 1; y = -1; }
else { x = 0; y = 1; }
if (y > 0)
  x = x + y; 
\end{verbatim}
The result for $x$ of the execution of this toy example, is always $x=1$, whatever $z \in [0,1]$.


Expliciting the dependency between $x$ and $y$ (we can always do so with constants) : 
in the two branches taken after the first test, we have $\hat y = -2\hat x + 1$, then we write
$\hat x = 0.5 + 0.5 \varepsilon_2$, and $\hat y= - \varepsilon_2$ after joining the results from 
the two branches. Then, interpreting test $(y > 0)$ leads to add constraint $\varepsilon_2 < 0$, 
but this is not enough to deduce $x=1$.

Then, we can note that for example $\hat x = 0.5 + 0.5 \varepsilon_2$ should denote a disjoint 
union of two values 1 and 0, thus in this case symbol $\varepsilon_2$ no longer takes all values 
in $[-1,1]$, but only the two values -1 and 1. Then, when test $(y > 0)$ is true, then  
$\varepsilon_2=-1$, $\hat x = 0.5 - 0.5 \varepsilon_2 = 1$, and when it is false,
 then $\hat x$ is naturally equal to 1.
\end{exem}


\section{Conclusion, Related and Future Work}

\label{conclusion}

We have {\em proved} that our abstract domain behaves well for an
interesting class of numerical programs. More work has yet to be done
on the formalisation of the shift operator (Section \ref{shiftoperator}) and
on more general schemes, such as some non-linear schemes of interest.
Many questions arise also from this work. For instance, can we replace
affine forms (but $\epsilon_U$) by higher-order Taylor models? 

Also, most of our proofs only rely
on the general properties of norms, and not specifically on $\ell_1$. 
What do we get with the $\ell_p$ norms, and in particular with the
standard Lorentz cone, when considering $\ell_2$? Many techniques are
available here that could help, in particular the techniques of {\em 
Second-Order Cone Programming}.

Our main convergence result (Theorem \ref{thm}) can be recast as a fixpoint
property of some general $(min,max,+)$ functions. Can we use 
policy iteration techniques \cite{CAV05,ESOP07} to help solve these? 
Last but not least, Property \ref{spectral} rings a bell and looks like
phenomena appearing with spectral measures (measures with value in 
a Banach space). Is this also generalizable to affine forms where
$\varepsilon_i$ are random variables of some sort? 

\paragraph{Acknowledgments} are due to St\'ephane Gaubert for 
interesting remarks during an earlier presentation of these results.

\end{document}